\documentclass[11pt,dvipsnames]{article}

\usepackage{fullpage}
\usepackage{mathptmx}
\usepackage{enumitem}
\usepackage{indentfirst}
\usepackage{amsfonts}
\usepackage{amsmath}
\usepackage{amsthm}
\usepackage{color}
\usepackage{verbatim}
\usepackage[makeroom]{cancel}
\usepackage{float}
\usepackage{amssymb}
\usepackage{amsthm}
\usepackage{mathrsfs}
\usepackage{epsfig}
\usepackage{graphicx}
\usepackage{indentfirst}
\usepackage{framed}
\usepackage{color}
\usepackage{hyperref}
\usepackage{graphicx}
\usepackage{float}
\usepackage{bm}
\usepackage[labelfont=bf]{caption}
\usepackage{tikz}
\usepackage{xcolor}
\usepackage{mathtools}

\definecolor{corlinks}{RGB}{0,0,220}
\definecolor{cormenu}{RGB}{15,93,219}
\definecolor{corurl}{RGB}{15,93,219}

\hypersetup{
	colorlinks=true,
	urlcolor=corlinks,
	linkcolor=corlinks,
	menucolor=cormenu,
	citecolor=corlinks,
	pdfborder= 0 0 0
}

\newtheorem{theorem}{Theorem}
\newtheorem{question}[theorem]{Question}

\newtheorem{lemma}[theorem]{Lemma}
\newtheorem{corollary}[theorem]{Corollary}
\newtheorem{definition}[theorem]{Definition}

\newtheorem{remark}[theorem]{Remark}

\newtheorem{fact}[theorem]{Fact}
\newtheorem{claim}[theorem]{Claim}

\newtheorem{example}[theorem]{Example}

\def\colorful{1}

\ifnum\colorful=1

\fi
\ifnum\colorful=0

\fi

\newcommand{\prob}[2]{\mathop{\mathrm{Pr}}_{#1}[#2]}
\newcommand{\avg}[2]{\mathop{\textbf{E}}_{#1}[#2]}
\newcommand{\poly}{\mathop{\mathrm{poly}}}

\newcommand{\F}{\mathbb{F}}

\newcommand{\AC}{\mathsf{AC}}

\newcommand{\mc}[1]{\mathcal{#1}}

\newcommand{\pdeg}{\mathrm{pdeg}}
\newcommand{\Maj}{\mathrm{Maj}}

\newcommand{\OR}{\mathrm{OR}}
\newcommand{\AND}{\mathrm{AND}}

\newcommand{\vertexcirc}[5]{
	\fill[white] (#1,#2) circle (0.35);
	\draw[#5] (#1,#2) circle (0.35);
	\node at (#1,#2) {#4#3};
}
\newcommand{\vertexsq}[5]{
	\fill[white] (#1-0.2,#2-0.2)--(#1+0.2,#2-0.2)--(#1+0.2,#2+0.2)--(#1-0.2,#2+0.2)--(#1-0.2,#2-0.2);
	\draw[#5] (#1-0.2,#2-0.2)--(#1+0.2,#2-0.2)--(#1+0.2,#2+0.2)--(#1-0.2,#2+0.2)--(#1-0.2,#2-0.2);
	\node at (#1,#2) {#4#3};
}


\title{On the Probabilistic Degree of an $n$-variate Boolean Function}
\author{Srikanth Srinivasan\thanks{On leave from Department of Mathematics, IIT Bombay. Supported by startup grant from Aarhus University.}\\
	Aarhus University, Denmark.\\
	\texttt{srikanth@cs.au.dk} 
	\and 
	S. Venkitesh\thanks{Department of Mathematics, IIT Bombay.  Supported by the Senior Research Fellowship of the Human Resource Development Group, Council of Scientific and Industrial Research, Government of India.}\\
	IIT Bombay, Mumbai, India. \\
	\texttt{venkitesh.mail@gmail.com}}
\date{}

\begin{document}
	\maketitle
	
	\begin{abstract}
		Nisan and Szegedy (CC 1994) showed that any Boolean function $f:\{0,1\}^n\rightarrow \{0,1\}$ that depends on all its input variables, when represented as a real-valued multivariate polynomial $P(x_1,\ldots,x_n)$, has degree at least $\log n - O(\log \log n)$. This was improved to a tight $(\log n - O(1))$ bound by Chiarelli, Hatami and Saks (Combinatorica 2020). Similar statements are also known for other Boolean function complexity measures such as Sensitivity (Simon (FCT 1983)), Quantum query complexity, and Approximate degree (Ambainis and de Wolf (CC 2014)). 
		
		In this paper, we address this question for \emph{Probabilistic degree}. The function $f$ has probabilistic degree at most $d$ if there is a random real-valued polynomial of degree at most $d$ that agrees with $f$ at each input with high probability. Our understanding of this complexity measure is significantly weaker than those above: for instance, we do not even know the probabilistic degree of the OR function, the best-known bounds put it between $(\log n)^{1/2-o(1)}$ and $O(\log n)$ (Beigel, Reingold, Spielman (STOC 1991); Tarui (TCS 1993); Harsha, Srinivasan (RSA 2019)).
		
		Here we can give a near-optimal understanding of the probabilistic degree of $n$-variate functions $f$, \emph{modulo} our lack of understanding of the probabilistic degree of OR. We show that if the probabilistic degree of OR is $(\log n)^c$, then the minimum possible probabilistic degree of such an $f$ is at least $(\log n)^{c/(c+1)-o(1)}$, and we show this is tight up to $(\log n)^{o(1)}$ factors.
	\end{abstract}
	
	\section{Introduction}
	\label{sec:intro}
	
	\subsection{Background and motivation}
	\label{sec:background}
	
	Representing Boolean functions $f:\{0,1\}^n\rightarrow \{0,1\}$ by polynomials is a tried-and-tested technique that has found uses in many areas of Theoretical Computer Science. In particular, such representations have led to important results in Complexity theory~\cite{Beigel, Braverman}, Learning theory~\cite{KlivansServedio,CIKK}, and Algorithm Design~\cite{WilliamsAPSP}. 
	
	There are many different kinds of  polynomial representations that are useful in various applications. The most straightforward way to represent a Boolean function $f:\{0,1\}^n\rightarrow \{0,1\}$ by a polynomial is by finding a $P\in \mathbb{R}[x_1,\ldots,x_n]$\footnote{We can represent $f$ as a polynomial over any field, but in this paper, we will work over the reals.} such that $P(a) = f(a)$ for all $a \in \{0,1\}^n$. It is a standard fact (say by M\"{o}bius Inversion or polynomial interpolation) that any $f$ has such a representation\footnote{The representation is in fact \emph{unique} if we restrict $P$ to be \emph{multilinear}, i.e. that no variable has degree more than $1$ in $f$.} has degree at most $n$, and the smallest degree of such a $P$ is called the \emph{degree of $f$} (or sometimes the \emph{Fourier degree of $f$} because of its close relation to the Fourier spectrum of $f$~\cite{ODonnellbook}), and denoted $\deg(f).$
	
	The degree of $f$ is an important notion of \emph{complexity} of the function $f$ and is closely related to a slew of combinatorial measures of Boolean function complexity such as \emph{Sensitivity, Decision Tree complexity, Quantum Query complexity}, etc.   (see, e.g., the survey of Buhrman and de Wolf~\cite{BdW} for a nice introduction). Given a complexity measure $\mu(\cdot)$ (such as $\deg(\cdot)$) on Boolean functions, a natural question to ask is the following. 
	
	\begin{question}
		\label{qn:main}
		How small can $\mu(f)$ be for a function $f$ on $n$ variables?
	\end{question}
	
	To make this question interesting, one must exclude trivial functions like the constant functions, and more generally, functions that depend on just a small subset of their input variables. This brings us to the following definition.
	
	\begin{definition}[Truly $n$-variate Boolean function\footnote{Such functions are also called \emph{non-degenerate} Boolean functions in the literature~\cite{Simon}. }]
		We say that a Boolean function $f(x_1,\ldots,x_n)$ depends on its input variable $x_i$, or equivalently that $x_i$ is \emph{influential} for $f$, if there is an input $a$ such that flipping the value of the $i$th variable at $a$ changes the value of $f$ (in this case, we also say that $x_i$ is influential for $f$ at $a$). We say that a Boolean function $f:\{0,1\}^n\rightarrow \{0,1\}$ is \emph{truly $n$-variate} if it depends on all its $n$ variables. 
	\end{definition}
	
	A number of results have addressed questions regarding how small complexity measures can be for truly $n$-variate Boolean functions.
	\begin{enumerate}[itemsep=2pt,parsep=2pt,topsep=1pt,labelindent=\parindent]
		\item Motivated by problems in Learning theory and PRAM lower bounds, Nisan and Szegedy~\cite{NS} showed that any truly $n$-variate function has degree at least $\log n - O(\log \log n).$ Recently, this was improved to $\log n - O(1)$ by Chiarelli, Hatami and Saks~\cite{CHS}. There are standard examples of Boolean functions (see, e.g., the \emph{Addressing function} defined below) for which this is tight.
		\item Ambainis and de Wolf~\cite{AdW} studied the same question for the \emph{approximate degree of $f$}, which is defined to be the minimum degree of  a polynomial $P$ such that $|P(a)-f(a)|<1/3$ for all $a\in \{0,1\}^n.$ This complexity measure is closely related to the \emph{quantum query complexity of $f$}~\cite{BdW}. 
		
		Ambainis and de Wolf~\cite{AdW} showed that any truly $n$-variate function has approximate degree (and also quantum query complexity) $\Omega(\log n/\log \log n).$ They also constructed variants of the Addressing function for which this bound is tight up to constant factors.
		
		\item Such results are also known for more combinatorial complexity measures, such as the \emph{sensitivity} of a Boolean function $f$, which is defined as follows. The sensitivity of $f$ at a point $a\in \{0,1\}^n$ is the number of input variables to $f$ that are influential for $f$ at $a$. The sensitivity of $f$ is the maximum sensitivity of $f$ at any input. 
		
		Simon~\cite{Simon} showed that any truly $n$-variate $f$ has sensitivity at least $\log n - O(\log \log n).$ This is also tight up to the $O(\log \log n)$ additive term (say, for the Addressing function).
	\end{enumerate}
	
	We address Question~\ref{qn:main} for another well-known polynomial-degree measure called the \emph{Probabilistic degree.} We define this notion first.
	
	\begin{definition}[Probabilistic polynomial and Probabilistic degree]
		Given a Boolean function $f:\{0,1\}^n\rightarrow \{0,1\}$ and an $\varepsilon\geq 0,$ an \emph{$\varepsilon$-error probabilistic polynomial} for $f$ is a random polynomial $\bm{P}$ (with some distribution having finite support) over $\mathbb{R}[x_1,\ldots,x_n]$\footnote{This can also be defined over other fields.} such that for each $a\in \{0,1\}^n$,
		\[
		\prob{\bm{P}}{\bm{P}(a) \neq f(a)} \leq \varepsilon.
		\]
		(Note that $\bm{P}(a)$ need not be Boolean when $\bm{P}(a)\neq f(a).$)
		
		We say that the degree of $\bm{P}$, denoted $\deg(\bm{P})$, is at most $d$ if the probability distribution defining $\bm{P}$ is supported on polynomials of degree at most $d$. Finally, we define the \emph{$\varepsilon$-error probabilistic degree} of $f$, denoted $\pdeg_\varepsilon(f)$, to be the least $d$ such that $f$ has an $\varepsilon$-error probabilistic polynomial of degree at most $d$.
		
		In the special case that $\varepsilon = 1/3,$ we omit the subscript in the notation above and simply use $\pdeg(f).$
	\end{definition}
	
	The probabilistic degree is a fundamentally important and well-studied complexity measure of Boolean functions. It was implicitly introduced (in the finite field setting) in a celebrated result of Razborov~\cite{Razborov}, who showed how to use it to construct low-degree polynomial approximations to small-depth circuits, and hence prove strong circuit lower bounds. The real-valued version was first studied by Beigel, Reingold and Spielman~\cite{BRS} and Tarui~\cite{Tarui} who were motivated by other circuit lower bound questions and oracle separations. This measure has since found other applications in complexity theory~\cite{ABFR, Beigel}, Pseudorandom generator constructions~\cite{Braverman}, Learning theory~\cite{CIKK}, and Algorithm design~\cite{WilliamsAPSP,AlmanCW16}. Further, in many of these applications (e.g, \cite{ABFR, Braverman, AlmanCW16}) we \emph{need} real-valued approximations.
	
	Despite this, however, our understanding of probabilistic degree is much less developed than the other measures above. For instance, near-optimal lower bounds of $n^{1-o(1)}$ on the probabilistic degree of an explicit Boolean function $f:\{0,1\}^n\rightarrow \{0,1\}$ were proved only recently by Viola~\cite{Viola}, and are only known for a function in the complexity class $\mathrm{E}^{\mathrm{NP}}$; in comparison, the Parity function has degree and approximate degree $n$, which is the largest possible. Another example is the OR function on $n$ variables. It is trivial to estimate the degree of OR (which is $n$) and well-known that its approximate degree is $\Theta(\sqrt{n})$~\cite{NS,Grover}. However, its  probabilistic degree (over the reals) remains unknown: the best known upper bound is $O(\log n)$ due to independent results of Beigel et al.~\cite{BRS} and Tarui~\cite{Tarui}, while the best lower bound is $(\log n)^{1/2-o(1)}$ due to Harsha and the first author~\cite{HS}. This indicates that we need better tools to understand probabilistic degree in general and over the reals in particular. This is one of the motivations behind this paper.
	
	Another motivation is to understand the contrast between the setting of real-valued probabilistic polynomials and polynomials over constant-sized finite fields. At a high level, this helps us understand the contrast between circuit complexity classes $\AC^0$ and $\AC^0[p]$, as the former class of circuits has low-degree probabilistic polynomials over the reals~\cite{BRS,Tarui}, while the latter does not~\cite{Smolensky}. It is easy to show that there are truly $n$-variate Boolean functions of constant degree over finite fields (e.g., the parity function is a linear polynomial over the field $\F_2$). It is interesting to ask to what extent such phenomena fail over the reals.
	
	A final motivating reason is to understand more precisely the relationships between probabilistic degree and other complexity measures such as approximate degree. A recent conjecture of Golovnev, Kulikov and Williams~\cite{GKW} shows that porting results for approximate degree to probabilistic degree would have interesting consequences for De Morgan formula lower bounds. By proving results such as the one in this paper, we hope to be able to prove such connections and hopefully uncover others.
	
	With these motivations in mind, we address Question~\ref{qn:main} in the setting of Probabilistic degree. That is, what is the lowest possible probabilistic degree of a truly $n$-variate Boolean function? As far as we know, this question has not been addressed before. Putting together Simon's bound on the sensitivity of a truly $n$-variate function with known probabilistic degree lower bounds~\cite{HS}, one can show a lower bound of $(\log \log n)^{1/2 - o(1)}.$ This is quite far from the best known upper bounds of $O(\log n)$, which hold for say the OR function~\cite{BRS, Tarui} and the Addressing function defined below in Section~\ref{sec:pf-outline}. 

	\subsection{Results}
	\label{sec:results}
	
	Our aim is to prove a result characterizing the minimum possible probabilistic degree of a truly $n$-variate Boolean function. However, the gap even just in our understanding of the OR function (as mentioned above) tells us that this may not yet be within reach. What we are able to do is to give a near-complete characterization modulo the gap between known upper and lower bounds for $\pdeg(\OR).$ \emph{Moreover, the answer is non-trivial: it is not simply $\pdeg(\OR).$}
	
	More precisely, our results are the following. Below, $\OR_n$ denotes the OR function on $n$ variables. We assume that we have bounds of the form $\pdeg(\OR_n) = (\log n)^{c\pm o(1)}$ for some $c>0$.
	
	\begin{theorem}
		\label{thm:lbd}
		Assume that $\pdeg(\OR_n) \geq (\log n)^{c-o(1)}$ for some $c > 0$ and all large enough $n\in \mathbb{N}$. Then, any truly $n$-variate Boolean function $f:\{0,1\}^n\rightarrow \{0,1\}$ satisfies $\pdeg(f) \geq (\log n)^{(c/(c+1))-o(1)}.$ 
	\end{theorem}
	
	\begin{theorem}
		\label{thm:ubd}
		Assume that $\pdeg(\OR_n) \leq (\log n)^{c+o(1)}$ for some $c > 0$ and all large enough $n\in \mathbb{N}$. Then, there exists a truly $n$-variate  Boolean function $f:\{0,1\}^n\rightarrow \{0,1\}$ such that $\pdeg(f) \leq (\log n)^{(c/(c+1))+o(1)}.$
	\end{theorem}
	
	Thus, we get close-to-matching lower and upper bounds for truly $n$-variate Boolean functions assuming close-to-matching lower and upper bounds for the OR function. However, the above statements also imply \emph{unconditional} lower and upper bounds on the probabilistic degrees of truly $n$-variate Boolean functions. Using known results that yield $(\log n)^{(1/2)-o(1)}\leq \pdeg(\OR_n)\leq O(\log n)$~\cite{BRS, Tarui, HS}, we get
	
	\begin{corollary}
		\label{cor:lbd}
		Any truly $n$-variate  Boolean function $f:\{0,1\}^n\rightarrow \{0,1\}$  satisfies $\pdeg(f) \geq (\log n)^{(1/3)-o(1)}.$ 
	\end{corollary}
	
	\begin{corollary}
		\label{cor:ubd}
		There exists a truly $n$-variate Boolean function $f:\{0,1\}^n\rightarrow \{0,1\}$ such that $\pdeg(f) \leq (\log n)^{(1/2)+o(1)}.$
	\end{corollary}
	
	\begin{remark}
		\label{rem:o(1)}
		The reader may wonder why we assume lower and upper bounds of the form $(\log n)^{c\pm o(1)}$ for $\pdeg(\OR_n)$. This is because the gaps between the known upper and lower bounds are $(\log n)^{\Omega(1)}$, and so it makes sense to use a characterization that shrinks this gap to something relatively insignificant. Furthermore, the best known lower bound on $\pdeg(\OR_n)$ is of the form $(\log n)^{1/2 - o(1)}$~\cite{HS} (more precisely, it is $\Omega((\log n)/(\log \log n)^{3/2})$).
		
		If we instead assume a more precise characterization $\pdeg(\OR_n) = \Theta((\log n)^c),$ then going through the proofs of the above theorems would yield a sharper lower bound of $\Omega((\log n)^{c/(c+1)}/(\log \log n)^{2})$ for any truly $n$-variate Boolean function and a better upper bound of $O((\log n)^{c/(c+1)})$ for some truly $n$-variate Boolean function.
	\end{remark}
	
	\subsection{Proof Outline}
	\label{sec:pf-outline}
	
	Our proof is motivated by two important examples. The first of these is the $\OR_n$ function which has probabilistic degree at most $O(\log n)$ by results of~\cite{BRS,Tarui} and at least $(\log n)^{(1/2)-o(1)}$ by~\cite{HS}. The second is the \emph{Addressing function}, which we now define. 
	
	The Addressing function $\mathrm{Addr}_r$ has $n= r+2^r$ variables. We think of the input variables as being divided into two parts: there are $r$ `addressing' variables $y_1,\ldots,y_r$ and $2^r$ `addressed' variables $\{z_a\ |\ a\in \{0,1\}^r\}$ (the latter part of the input is thus indexed by elements of $\{0,1\}^r$). On an input $(a,A)\in \{0,1\}^{r}\times \{0,1\}^{2^r},$ the output of the function is defined to be $A_a$ (i.e. the $a$th co-ordinate of the vector $A$). The Addressing function satisfies $\deg(\mathrm{Addr}_r) = r+1 = O(\log n).$ This example is quite relevant to this line of work: in particular, it implies that the results of Nisan and Szegedy~\cite{NS} and Chiarelli et al.~\cite{CHS} stated above are tight, and is also a tight example for Simon's theorem~\cite{Simon}.
	
	We now describe the upper and lower bound proofs, starting with the less technical upper bound.
	
	\paragraph{The Upper Bound.} Given that we have two natural families of truly $n$-variate functions that have degree $O(\log n),$ one may suspect that this is the best possible. Indeed this was also our initial conjecture. However, using the ideas of Ambainis and de Wolf~\cite{AdW}, we can do better.  Ambainis and de Wolf showed that there are truly $n$-variate Boolean functions that have \emph{approximate degree} $O(\log n/\log \log n).$ Their construction\footnote{They actually give two, slightly different, constructions. We use the second one here.} uses a modified Addressing function, where the addressing variables are present in an `encoded' form. While this blows up the size of the first part of the input, this does not affect $n$ much as the addressing variables take up only a small part of the input. On the other hand, the advantage is that the `decoding' procedure can be performed approximately by a suitable low-degree polynomial: a proof of this uses two famous Quantum algorithms, the \emph{Bernstein-Vazirani algorithm} and \emph{Grover search}, along with the fact that efficient Quantum algorithms yield approximating low-degree polynomials~\cite{BBCMW}. Putting things together yields an improved approximate degree bound for some $n$-variate $f$.
	
	We show how to port their construction to the probabilistic degree setting. The first observation is that Grover search, which is essentially an algorithm for computing $\OR_n$, is much more `efficient' in the probabilistic degree setting, as $\pdeg(\OR_n) = O(\log n)$, while its approximate degree is $\Omega(\sqrt{n})$~\cite{NS}. The second observation is that the Bernstein-Vazirani algorithm, which can be thought of as a decoding algorithm for a suitable error-correcting code, can be replaced by polynomial interpolation. This gives a good idea of why we should also be able to use a similar construction in the probabilistic degree setting. In fact, the better probabilistic degree upper bound for $\OR_n$ implies that we should be able to get a \emph{better} bound than what is possible for approximate degree. Indeed this is true. By a similar construction, we show that we can construct a truly $n$-variate $f$ with probabilistic degree $O(\sqrt{\log n})$ \emph{unconditionally}, which is quite a bit better than previous results for any of the above degree measures. If we assume, moreover, that $\pdeg(\OR_n) \leq (\log n)^{c+o(1)},$ the same construction yields a function with probabilistic degree $(\log n)^{(c/(c+1))+o(1) }$.
	
	\paragraph{The Lower Bound.} Given that the upper bound construction uses the Addressing function as well as the OR function, it is only natural that the lower bound would use the lower bounds for these two families of functions. Our hypothesis already assumes a lower bound of $(\log n)^{c-o(1)}$ for $\pdeg(\OR_n).$ For the addressing function $\mathrm{Addr}_r$ described above, one can prove an $\Omega(r) = \Omega(\log n)$ lower bound in the following way. We observe that by setting the $2^r$ addressed variables uniformly at random, we obtain a uniformly random function on the $r$ addressing variables. By a counting argument, one can show that a uniformly random Boolean function $\bm{F}$ on $r$ variables has probabilistic degree $\Omega(r)$ with high probability. In particular, as setting some input variables to constants can only reduce probabilistic degree, this implies that $\pdeg(\mathrm{Addr}_r) = \Omega(r) = \Omega(\log n)$. Note that this is tight, as $\deg(\mathrm{Addr}_r) = r+1.$
	
	Our aim is to generalize the above lower bounds enough to prove a lower bound for any truly $n$-variate $f$. The first informal observation is that the $\OR_n$ function is the `simplest' function on $n$ variables to have sensitivity $n$. Therefore, it is intuitive that any Boolean function with sensitivity $n$ should have probabilistic degree at least that of the $\OR_n$ function. We show that this is true, up to $(\log n)^{o(1)}$ factors. More generally, we show that any Boolean function $f$ with sensitivity $s$ has probabilistic degree at least that of the OR function on $s$ variables (up to $(\log s)^{o(1)}$ factors). The proof of this is in the contrapositive: we use a probabilistic degree upper bound for $f$ to construct a probabilistic polynomial for $\OR_s$. The ideas behind this go back to a sampling argument used in the works of Beigel et al. and Tarui~\cite{BRS,Tarui}. Viewing this argument more abstractly, we can use this to construct a reduction from $\OR_s$ to $f$ (for any $f$ of sensitivity $s$) in the probabilistic degree setting. 
	
	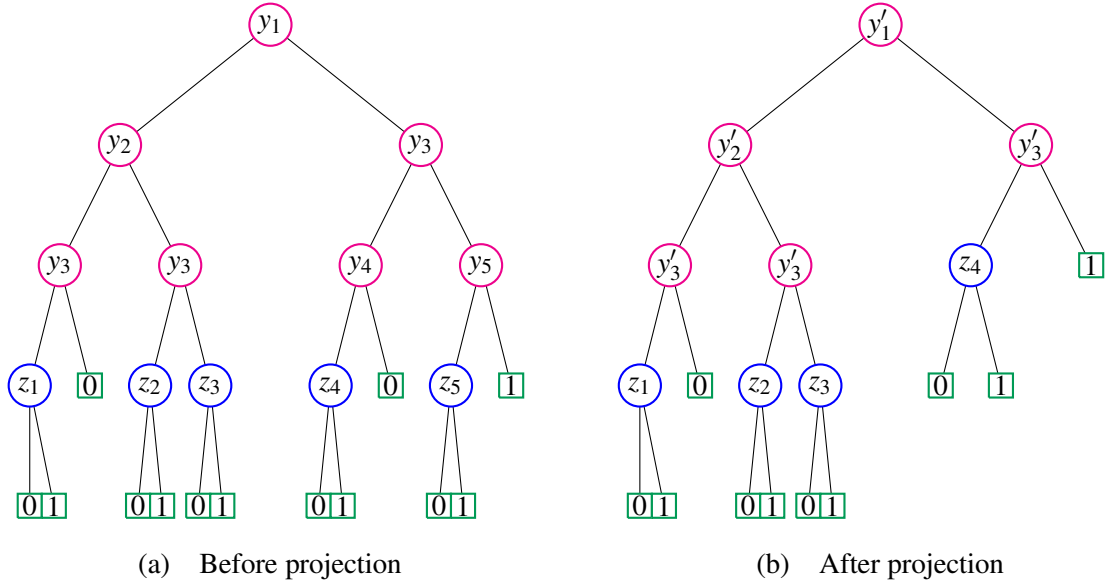
\begin{figure}[t!]
		\centering
		\begin{tikzpicture}[scale=0.8]
			\draw (0,0)--(-5+2.5,-2);
			\draw (0,0)--(5-2.5,-2);
			
			\draw (-5+2.5,-2)--(-6+2.5,-4);
			\draw (-5+2.5,-2)--(-4+2.5,-4);
			\draw (5-2.5,-2)--(4-2.5,-4);
			\draw (5-2.5,-2)--(6-2.5,-4);
			
			\draw (-7+1+2.5,-4)--(-7.5+1+2.5,-6);
			\draw (-7+1+2.5,-4)--(-6-0.5+1+2.5,-6);
			\draw (-3-1+2.5,-4)--(-4+0.5-1+2.5,-6);
			\draw (-3-1+2.5,-4)--(-2-0.5-1+2.5,-6);
			\draw (3+1-2.5,-4)--(2+0.5+1-2.5,-6);
			\draw (3+1-2.5,-4)--(4-0.5+1-2.5,-6);
			\draw (7-1-2.5,-4)--(6+0.5-1-2.5,-6);
			\draw (7-1-2.5,-4)--(7.5-1-2.5,-6);
			
			\draw (-7.5+1+2.5,-6)--(-7.5+1+2.5,-8);
			\draw (-7.5+1+2.5,-6)--(-6.5-0.4-0.2+1+2.5,-8);
			\draw (-4+0.5-1+2.5,-6)--(-4.5+0.2+0.5+0.1-1+2.5,-8);
			\draw (-4+0.5-1+2.5,-6)--(-3.5-0.2+0.5-0.1-1+2.5,-8);
			\draw (-2-0.5-1+2.5,-6)--(-2.5+0.2-0.5+0.1-1+2.5,-8);
			\draw (-2-0.5-1+2.5,-6)--(-1.5-0.2-0.5-0.1-1+2.5,-8);
			\draw (2+0.5+1-2.5,-6)--(1.5+0.2+0.5+0.1+1-2.5,-8);
			\draw (2+0.5+1-2.5,-6)--(2.5-0.2+0.5-0.1+1-2.5,-8);
			\draw (6+0.5-1-2.5,-6)--(5.5+0.2+0.5+0.1-1-2.5,-8);
			\draw (6+0.5-1-2.5,-6)--(6.5-0.2+0.5-0.1-1-2.5,-8);

			\vertexcirc{0}{0}{\(y_1\)}{}{Magenta,thick}
			
			\vertexcirc{-5+2.5}{-2}{\(y_2\)}{}{Magenta,thick}
			\vertexcirc{5-2.5}{-2}{\(y_3\)}{}{Magenta,thick}
			
			\vertexcirc{-6+2.5}{-4}{\(y_3\)}{}{Magenta,thick}
			\vertexcirc{-4+2.5}{-4}{\(y_3\)}{}{Magenta,thick}
			\vertexcirc{4-2.5}{-4}{\(y_4\)}{}{Magenta,thick}
			\vertexcirc{6-2.5}{-4}{\(y_5\)}{}{Magenta,thick}
			
			\vertexcirc{-7.5+1+2.5}{-6}{\(z_1\)}{}{blue,thick}
			\vertexsq{-6.5+1+2.5}{-6}{0}{}{ForestGreen,thick}
			\vertexcirc{-3.5-1+2.5}{-6}{\(z_2\)}{}{blue,thick}
			\vertexcirc{-2.5-1+2.5}{-6}{\(z_3\)}{}{blue,thick}
			\vertexcirc{2.5+1-2.5}{-6}{\(z_4\)}{}{blue,thick}
			\vertexsq{3.5+1-2.5}{-6}{0}{}{ForestGreen,thick}
			\vertexcirc{6.5-1-2.5}{-6}{\(z_5\)}{}{blue,thick}
			\vertexsq{7.5-1-2.5}{-6}{1}{}{ForestGreen,thick}
			
			\vertexsq{-7.5+1+2.5}{-8}{0}{}{ForestGreen,thick}
			\vertexsq{-6.5-0.4-0.2+1+2.5}{-8}{1}{}{ForestGreen,thick}
			\vertexsq{-4.5+0.2+0.5+0.1-1+2.5}{-8}{0}{}{ForestGreen,thick}
			\vertexsq{-3.5-0.2+0.5-0.1-1+2.5}{-8}{1}{}{ForestGreen,thick}
			\vertexsq{-2.5+0.2-0.5+0.1-1+2.5}{-8}{0}{}{ForestGreen,thick}
			\vertexsq{-1.5-0.2-0.5-0.1-1+2.5}{-8}{1}{}{ForestGreen,thick}
			\vertexsq{1.5+0.2+0.5+0.1+1-2.5}{-8}{0}{}{ForestGreen,thick}
			\vertexsq{2.5-0.2+0.5-0.1+1-2.5}{-8}{1}{}{ForestGreen,thick}
			\vertexsq{5.5+0.2+0.5+0.1-1-2.5}{-8}{0}{}{ForestGreen,thick}
			\vertexsq{6.5-0.2+0.5-0.1-1-2.5}{-8}{1}{}{ForestGreen,thick}
			
			\node at (0,-9) {(a)\quad Before projection};
		\end{tikzpicture}
		\hspace{1cm}
		\begin{tikzpicture}[scale=0.8]
			\draw (0,0)--(-5+2.5,-2);
			\draw (0,0)--(5-2.5,-2);
			
			\draw (-5+2.5,-2)--(-6+2.5,-4);
			\draw (-5+2.5,-2)--(-4+2.5,-4);
			\draw (5-2.5,-2)--(4-2.5,-4);
			\draw (5-2.5,-2)--(6-2.5,-4);
			
			\draw (-7+1+2.5,-4)--(-7.5+1+2.5,-6);
			\draw (-7+1+2.5,-4)--(-6-0.5+1+2.5,-6);
			\draw (-3-1+2.5,-4)--(-4+0.5-1+2.5,-6);
			\draw (-3-1+2.5,-4)--(-2-0.5-1+2.5,-6);
			\draw (3+1-2.5,-4)--(2+0.5+1-2.5,-6);
			\draw (3+1-2.5,-4)--(4-0.5+1-2.5,-6);
			
			\draw (-7.5+1+2.5,-6)--(-7.5+1+2.5,-8);
			\draw (-7.5+1+2.5,-6)--(-6.5-0.4-0.2+1+2.5,-8);
			\draw (-4+0.5-1+2.5,-6)--(-4.5+0.2+0.5+0.1-1+2.5,-8);
			\draw (-4+0.5-1+2.5,-6)--(-3.5-0.2+0.5-0.1-1+2.5,-8);
			\draw (-2-0.5-1+2.5,-6)--(-2.5+0.2-0.5+0.1-1+2.5,-8);
			\draw (-2-0.5-1+2.5,-6)--(-1.5-0.2-0.5-0.1-1+2.5,-8);

			\vertexcirc{0}{0}{\(y'_1\)}{}{Magenta,thick}
			
			\vertexcirc{-5+2.5}{-2}{\(y'_2\)}{}{Magenta,thick}
			\vertexcirc{5-2.5}{-2}{\(y'_3\)}{}{Magenta,thick}
			
			\vertexcirc{-6+2.5}{-4}{\(y'_3\)}{}{Magenta,thick}
			\vertexcirc{-4+2.5}{-4}{\(y'_3\)}{}{Magenta,thick}
			\vertexcirc{4-2.5}{-4}{\(z_4\)}{}{blue,thick}
			\vertexsq{6-2.5}{-4}{1}{}{ForestGreen,thick}
			
			\vertexcirc{-7.5+1+2.5}{-6}{\(z_1\)}{}{blue,thick}
			\vertexsq{-6.5+1+2.5}{-6}{0}{}{ForestGreen,thick}
			\vertexcirc{-3.5-1+2.5}{-6}{\(z_2\)}{}{blue,thick}
			\vertexcirc{-2.5-1+2.5}{-6}{\(z_3\)}{}{blue,thick}
			\vertexsq{2.5+1-2.5}{-6}{0}{}{ForestGreen,thick}
			\vertexsq{3.5+1-2.5}{-6}{1}{}{ForestGreen,thick}
			
			\vertexsq{-7.5+1+2.5}{-8}{0}{}{ForestGreen,thick}
			\vertexsq{-6.5-0.4-0.2+1+2.5}{-8}{1}{}{ForestGreen,thick}
			\vertexsq{-4.5+0.2+0.5+0.1-1+2.5}{-8}{0}{}{ForestGreen,thick}
			\vertexsq{-3.5-0.2+0.5-0.1-1+2.5}{-8}{1}{}{ForestGreen,thick}
			\vertexsq{-2.5+0.2-0.5+0.1-1+2.5}{-8}{0}{}{ForestGreen,thick}
			\vertexsq{-1.5-0.2-0.5-0.1-1+2.5}{-8}{1}{}{ForestGreen,thick}
			
			\node at (0,-9) {(b)\quad After projection};
		\end{tikzpicture}
		\caption{\small The function $f(y_1,\ldots,y_5,z_1,\ldots,z_5)$ is defined by the decision tree on the left (we assume that the left child corresponds to the queried variable taking value $0$). When $z_1,\ldots,z_5$ are set i.u.a.r. to $\bm{b}_1,\ldots,\bm{b}_5$, we get a random function $\bm{F}(y_1,\ldots,y_5)$ such that $\bm{F}(00000) = \bm{b}_1, \bm{F}(01000) = \bm{b}_2, \bm{F}(01100) = \bm{b}_3, \bm{F}(10000) = \bm{b}_4, \bm{F}(10100) = \bm{b}_5.$ After a projection that maps $y_1\mapsto y_1'; y_2 \mapsto y_2'; y_3,y_4,y_5\mapsto y_3'$, we get the function $f'$ computed by the tree on the right. This reduces the number of addressing variables to $3$. But also note that the variable $z_5$ is no longer relevant as the path leading to it is inconsistent with the projection. So the number of addressed variables falls to $4$.} 
		\label{fig:outline}
	\end{figure}
	
	The above argument implies a strong lower bound for any $n$-variate $f$ with large sensitivity. In particular, it implies that if $f$ has sensitivity at least $s = n^{\Omega(1)}$, then its probabilistic degree is almost that of the OR function. We now consider the case of functions with small sensitivity (specifically when $s = n^{o(1)}$), which is the most technical part of the proof. By a recent breakthrough result of Huang~\cite{Huang}, we also know that $f$ also has a \emph{decision tree}  (we refer the  reader to~\cite{BdW} for the definition of Decision trees) of depth $d = \poly(s) = n^{o(1)}.$\footnote{Strictly speaking, we do not need to use Huang's result as we could also use the known polynomial relationship between the decision tree height and the \emph{block sensitivity} of a function~\cite{Nisan-bs}. But it is notationally easier to work with sensitivity.} The prototypical example of such an $f$ is the Addressing function which has a decision tree of depth $r +1 =\lfloor \log n\rfloor + 1,$ which we argued a lower bound for above. The idea, in general, is to find a copy of something like an Addressing function `inside' the function $f$.
	
	We illustrate how this argument works by considering a special case of the problem, which is only a small variant of the Addressing function. Assume that a truly $n$-variate function $f$ is computed by a decision tree $T$ of depth $d = \poly(\log n)$. Note that as the function depends on all its variables, each of the underlying $n$ variables appear in the tree $T$. To make things even simpler, assume that we have $n/2$ `addressing' variables $y_1,\ldots,y_{n/2}$ and $n/2$ `addressed' variables $z_1,\ldots,z_{n/2}.$ The tree reads $d-1$ addressing variables among $y_1,\ldots,y_{n/2}$ in some (possibly adaptive) fashion and then possibly queries one addressed variable, the value of which is output. (See Figure~\ref{fig:outline} (a).)
	
	How do we argue a lower bound on $\pdeg(f)$? We could try to proceed as above and set the addressed variables $z_1,\ldots,z_{n/2}$ as random to obtain a random function $\bm{F}$ in $y_1,\ldots,y_{n/2}.$ However, this function is not \emph{uniformly} random, as it is sampled using only $n/2$ random bits, while the number of functions in $n/2$ variables is $2^{2^{n/2}}.$ Nevertheless, we can observe that the function $\bm{F}$ does take independent random values at at least $n/2$ distinct inputs, those which are consistent with $n/2$ distinct paths in $T$ leading to the various addressed variables. (See Figure~\ref{fig:outline} (a).) We could try to lower bound $\pdeg(\bm{F})$ as above.
	
	This leads to the following general question: given a random function $\bm{F}:\{0,1\}^r\rightarrow \{0,1\}$ that takes independent and random values at $M$ distinct inputs in $\{0,1\}^r$, what can we say about the probabilistic degree of $\bm{F}$? By a more general counting argument, we are able to show that with high probability, the probabilistic degree of $\bm{F}$ is at least $\Omega(\log M/\log r).$ This is easily seen to be tight in the case that $X$ is, say, a Hamming ball of radius $R\leq r^{1-\Omega(1)}$.  (In the case that $M=2^r,$ this leads to a bound of $\Omega(r/\log r)$, nearly matching the claim for random functions that we mentioned above. A tight bound can be obtained in the same way but is harder to state for general $M$.)
	
	Given this bound for random functions, we can try to use it in the case of the function $f$ above. Unfortunately, in this case, both parameters $r$ and $M$ are $n/2$, and hence we do not get any non-trivial bound. However, we show that we can still reduce to a case where a non-trivial bound is possible (this is where the depth of $T$ comes in). More precisely, we  reduce the number of addressing variables by \emph{projecting} the $n/2$ addressing variables to a smaller set of $r'$ variables $Y' = \{y'_1,\ldots,y'_{r'}\}.$ That is, we randomly set each variable $Y$ to a uniformly random variable in $Y'$ to get a different function in the variables $Y'\cup Z$. This has the effect of reducing the number of addressing variables to $r'$. But there is also a potential problem: the projection could also render some of the addressed variables irrelevant, as the paths that lead to them become inconsistent. (See Figure~\ref{fig:outline} (b).)
	
	Nevertheless, if we choose $r'$ large enough (something like $r' = 4d^2$ is enough by the Birthday paradox), the variables of each path are sent to \emph{distinct} variables in $Y'$ with high probability, which implies that each addressed variable remains relevant with high probability. In particular, there is a projection that maps $f$ to an `Addressing function' with only $\poly(\log n)$ addressing variables and $\Omega(n)$ addressed variables. Now applying the argument for random functions, we get a probabilistic degree lower bound of $\Omega(\log n/\log \log n)$ for this function, nearly matching what we obtained for the Addressing function. As projections do not increase probabilistic degree, the same bound holds for $f$, concluding the proof in this special case.\footnote{Random projections of this kind have been used recently to prove important results in circuit complexity~\cite{HRST,COST}. However, as far as we know, they have not been used to prove probabilistic degree lower bounds.}
	
	A similar argument can be carried out in the general case by first carefully partitioning the variables into the addressing and addressed variables. We do this by looking at the structure of the decision tree $T$. These details are postponed to the formal proof. In general, this argument yields a lower bound of $\Omega(\log n/\log s)$ on the probabilistic degree of a truly $n$-variate function $f$ with sensitivity at most $s$. 
	
	Using this lower bound along with the previous lower bound for functions of sensitivity at least $s$, and optimizing our choice of $s$, yields a lower bound of $(\log n)^{c/(c+1)-o(1)}$ for any truly $n$-variate function $f$.
	
	\section{Preliminaries}
	
	\paragraph{Functions, Restrictions, Projections.} Throughout, we work with real-valued functions $f:\{0,1\}^n\rightarrow \mathbb{R}$. Boolean functions (i.e. functions mapping $\{0,1\}^n$ to $\{0,1\}$) are also treated as real-valued. We use boldface notation to denote random variables. A \emph{random function} $\bm{F}$ is a probability distribution over functions. 
	
	A \emph{restriction} on $n$ variables is a map $\rho:[n]\rightarrow \{0,1,*\}$. Given a function $f:\{0,1\}^n\rightarrow \mathbb{R}$ and a restriction $\rho$ on $n$ variables, we have a natural restricted function $f_\rho$ defined by setting the $i$th input variable to $f$ to $0$, $1$ or  leaving it as is, depending on whether $\rho(i)$ is $0$, $1$ or $*$ respectively. Note that the function $f_\rho$ now depends on $|\rho^{-1}(*)|$ many variables. However, we sometimes also treat $f_\rho$ as  a function of all the original variables that only \emph{depends} on (a subset of) the variables indexed by $\rho^{-1}(*).$
	
	A \emph{projection} from $n$ variables to $m$ variables is a map $\nu:[n] \rightarrow [m]$. Given a function $f:\{0,1\}^n\rightarrow \mathbb{R}$ and a projection $\nu$ from $n$ variables to $m$ variables, we get a function $f|_\nu:\{0,1\}^m\rightarrow \mathbb{R}$ by identifying variables of $f$ that map to the same image under $\nu$.
	
	\paragraph{Some Boolean functions.} For any positive integer $n$, we use $\OR_n, \AND_n$ and $\Maj_n$ to denote the OR, AND and Majority functions on $n$ variables respectively.
	
	\begin{fact}
		\label{fac:pdeg}
		We have the following simple facts about probabilistic polynomials.
		\begin{enumerate}[itemsep=2pt,parsep=2pt,topsep=1pt,labelindent=\parindent]
			\item (Interpolation) Any function $f:\{0,1\}^n\rightarrow \{0,1\}$ has an exact multilinear polynomial representation of degree at most $n$. I.e. $\deg(f) := \pdeg_0(f) \leq n.$
			\item (Shifts and Restrictions) Fix any $f:\{0,1\}^n\rightarrow \{0,1\}$ and any $\varepsilon \geq 0$. Then the function $g:\{0,1\}^n\rightarrow \{0,1\}$ defined by $g(x) = f(x\oplus y)$ for a fixed $y\in \{0,1\}^n$ has the same probabilistic degree as $f$, i.e., $\pdeg_\varepsilon(g) = \pdeg_\varepsilon(f)$. 
			
			If $g:\{0,1\}^m\rightarrow \{0,1\}$ is a restriction or a projection of $f$, then $\pdeg_\varepsilon(g) \leq \pdeg_\varepsilon(f).$
			\item (Error reduction~\cite{HS}) For any $\delta < \varepsilon \leq 1/3$ and any Boolean function $f$, if $\bm{P}$ is an $\varepsilon$-error probabilistic polynomial for $f$, then $\bm{Q} = M(\bm{P}_1,\ldots,\bm{P}_\ell)$ is a $\delta$-error probabilistic polynomial for $f$ where \(\ell=O(\log(1/\delta)/\log(1/\varepsilon))\), $M$ is the exact multilinear polynomial for $\Maj_\ell$ and $\bm{P}_1,\ldots,\bm{P}_\ell$ are independent copies of $\bm{P}.$ In particular, we have $\pdeg_{\delta}(f) \leq \pdeg_{\varepsilon}(f)\cdot O(\log(1/\delta)/\log(1/\varepsilon)).$ 
			\item (Composition) For any Boolean function $f$ on $k$ variables and any Boolean functions $g_1,\ldots,g_k$ on a common set of $m$ variables,  let $h$ denote the natural composed function $f(g_1,\ldots,g_k)$ on $m$ variables. For $\varepsilon, \delta_1,\ldots,\delta_k \geq 0$, let $\bm{P}, \bm{Q}_1,\ldots,\bm{Q}_k$ be probabilistic polynomials for $f,g_1,\ldots,g_k$ respectively with errors $\varepsilon, \delta_1,\ldots,\delta_k$ respectively. Then, $\bm{R} = \bm{P}(\bm{Q}_1,\ldots,\bm{Q}_k)$ is a probabilistic polynomial for $h$ with error at most $\varepsilon + \sum_i \delta_i.$
			
			In particular, for any $\varepsilon, \delta > 0,$ we have $\pdeg_{\varepsilon + k\delta}(h) \leq \pdeg_\varepsilon(f)\cdot \max_{i\in [k]} \pdeg_\delta(g_i).$
		\end{enumerate}
	\end{fact}
	
	We will need the following known upper and lower bounds on $\pdeg(\OR_n).$
	
	\begin{theorem}[\cite{BRS,Tarui}]
		\label{thm:BRST}
		$\pdeg_\varepsilon(\OR_n) = O(\log n \log(1/\varepsilon)).$
	\end{theorem}
	
	\begin{theorem}[\cite{HS}]
		\label{thm:HS}
		$\pdeg(\OR_n) \geq (\log n)^{1/2-o(1)}.$
	\end{theorem}

	\begin{definition}[Some Complexity Measures of Boolean functions]
		Let $f:\{0,1\}^n\rightarrow \{0,1\}$ be any Boolean function. We use $D(f)$ to denote the depth of the smallest Decision Tree computing $f$.
		
		For $a\in \{0,1\}^n$, we use $s(f,a)$ to denote the number of $b\in \{0,1\}^n$ that can be obtained by flipping a single bit of $a$ and satisfying $f(a) \neq f(b).$ The \emph{Sensitivity} of $f$, denoted $s(f),$ is defined to be the maximum value of $s(f,a)$ as $a$ ranges over  $\{0,1\}^n.$
	\end{definition}
	
	Huang~\cite{Huang} proved the following breakthrough result recently.
	
	\begin{theorem}[Huang's Sensitivity theorem~\cite{Huang}]
		\label{thm:Huang}
		There is an absolute constant $c_0 > 0$ such that for all large enough $n$ and all functions $f:\{0,1\}^n\rightarrow \{0,1\}$, $D(f)\leq s(f)^{c_0}.$
	\end{theorem}
	
	Strictly speaking, we do not need to use Huang's Sensitivity theorem in what follows as we could also make do with a polynomial relationship between the decision tree height and the \emph{block sensitivity}\footnote{The Block sensitivity of a function $f:\{0,1\}^n\rightarrow \{0,1\}$ is defined as follows. Given $a\in \{0,1\}^n$, define  $bs(f,a)$ to be the maximum number of pairwise disjoint  sets $B_1,\ldots,B_t\subseteq [n]$ such that flipping all the bits indexed by any $B_i$ in $a$ results in an input $b^{(i)}$ such that $f(a)\neq f(b^{(i)})$. Then, the block sensitivity of $f$ is defined to be the maximum value of $bs(f,a)$ over all inputs $a\in \{0,1\}^n$.} of $f$, which has been known for a long time~\cite{Nisan-bs}. However, it is notationally simpler to work with sensitivity.

	\section{The Lower Bound: Proof of Theorem~\ref{thm:lbd}}
	
	The proof is made up of two lower bounds. We first prove a lower bound on $\pdeg(f)$ for any function $f$ that has large sensitivity $s(f)$; this is by a suitable reduction from the case of the $\OR$ function. We then prove a lower bound on $\pdeg(f)$ for any function that depends on all its variables but has small sensitivity; this is by a suitable reduction from a kind of Addressing function. Optimizing over the parameters of the lower bounds will yield the lower bound of the theorem statement.
	
	Throughout this section, we assume that $\pdeg(\OR_n) \geq (\log n)^{c-o(1)}$ for all large enough $n$.
	
	\subsection{The case of large sensitivity}
	
	The main result of this section is the following lower bound on the probabilistic degrees of Boolean functions with large sensitivity.
	
	\begin{lemma}
		\label{lemma:lbd-high-sens}
		Let $f:\{0,1\}^n\rightarrow \{0,1\}$ be any Boolean function that has sensitivity $s$. Then, $\pdeg(f) \geq (\log s)^{c - o(1)}.$
	\end{lemma}
	
	The above lemma is proved via a probabilistic reduction from the $\OR$ function on $s$ variables to the function $f$. This is captured by the following lemma, which shows how a function that has large sensitivity can be used to obtain a probabilistic representation of a large copy of the $\OR$ function. 
	
	Recall from above that a Boolean function $h:\{0,1\}^s\rightarrow \{0,1\}$ is a restriction of a Boolean function $g:\{0,1\}^s\rightarrow \{0,1\}$ if $h$ can be obtained by setting  some inputs of $g$ to constants. Though $h$ no longer depends on the variables that are set to constants, here we still treat $h$ as a function on all $s$ variables.
	
	\begin{lemma}
		\label{lem:OR-redn}
		Let $g:\{0,1\}^s \rightarrow \{0,1\}$ be any Boolean function such that $g(0^s)  = 0$ and $g(x) = 1$ for any $x$ of Hamming weight $1$. Then, there exist $\ell = O(\log s)$ independent random restrictions $\bm{g}_1,\ldots,\bm{g}_\ell$ of $g$ such that for any $a\in \{0,1\}^s$,
		\[
		\prob{\bm{g}_1,\ldots,\bm{g}_\ell}{\OR_\ell(\bm{g}_1(a),\ldots,\bm{g}_\ell(a)) \neq \OR_s(a)} \leq\frac{1}{10}.
		\]
	\end{lemma}
	
	We interpret the random function $\OR_\ell(\bm{g}_1(a),\ldots,\bm{g}_\ell(a))$ as a probabilistic representation of the $\OR_s$ function. The reader may be confused by the fact that the probabilistic representation itself uses an $\OR$ function; however, note that this $\OR$ function is defined on $\ell \ll s$ variables and consequently is a much `simpler' function (in particular, for us, what is relevant is that $\pdeg(\OR_\ell) = O(\log \ell)$~\cite{BRS,Tarui} which is much smaller than $\log s$) . 
	
	The proof of Lemma~\ref{lem:OR-redn} is closely related to the argument for constructing a probabilistic polynomial for the OR function from~\cite{BRS, Tarui}. The observation here is that a similar argument can be used to give a probabilistic reduction from $\OR_s$ to any function $g$ as above.
	
	Assuming Lemma~\ref{lem:OR-redn} for now, we first finish the proof of Lemma~\ref{lemma:lbd-high-sens}.
	
	\begin{proof}[Proof of Lemma~\ref{lemma:lbd-high-sens}]
		We know that $f$ has some input of sensitivity $s$.  Then we note that we may assume \(f(0^n)=0\) and \(f(0^{j-1}10^{n-j})=1\) for \(j\in[s]\).  For let \(a\in\{0,1\}^n\) such that \(s(f,a)=s\).  If \(f(a)=1\), we may replace \(f\) by \(1-f\).  (Obviously, \(\pdeg_\varepsilon(f)=\pdeg_\varepsilon(1-f)\), for all \(\varepsilon\ge0\).)  So we may assume \(f(a)=0\).  Now by permuting coordinates if required, we may assume that \(f(\tilde{a}^{(j)})=1\), where \(\tilde{a}^{(j)}\coloneqq(a_1,\ldots,a_{j-1},1-a_j,a_{j+1},\ldots,a_n)\) for all \(j\in[s]\).  Further, if \(a\ne0^n\), we may replace \(f\) by \(f'\), defined as \(f'(x)=f(x\oplus a),\,x\in\{0,1\}^n\).  By Fact~\ref{fac:pdeg} Item 2, \(\pdeg_\varepsilon(f)=\pdeg_\varepsilon(f')\), for all \(\varepsilon\ge0\).  Clearly, we have \(f'(0^n)=0\) and \(f'(0^{j-1}10^{n-j})=1\), for all \(j\in[s]\).
		
		So now, by assumption, we have \(f(0^n)=0\) and \(f(0^{j-1}10^{n-j})=1\) for \(j\in[s]\).  Define \(g:\{0,1\}^s\to\{0,1\}\) as \(g(x)=f(x0^{n-s})\).  Then $g$ satisfies the hypotheses of Lemma~\ref{lem:OR-redn}.  Hence, by Lemma~\ref{lem:OR-redn}, there exist \(\ell=O(\log s)\) random restrictions \(\bm{g}_1,\ldots,\bm{g}_\ell\) of \(g\) such that
		\begin{equation}
			\label{eq:prob-rep}
			\prob{\bm{g}_1,\ldots,\bm{g}_\ell}{\OR_\ell(\bm{g}_1(x),\ldots,\bm{g}_\ell(x))\ne\OR_s(x)}\le\frac{1}{10},\quad\text{for all }x\in\{0,1\}^s.
		\end{equation}
		We use the above representation to devise a probabilistic polynomial for $\OR_s$. 
		
		Let \(\bm{O}\) be any \((1/10)\)-error probabilistic polynomial for \(\OR_\ell\) and $\bm{G}_1,\ldots,\bm{G}_\ell$ be any \((1/10\ell)\)-error probabilistic polynomials for \(\bm{g}_1,\ldots,\bm{g}_\ell\) respectively. Then, by Fact~\ref{fac:pdeg} and (\ref{eq:prob-rep}), \(\bm{O}(\bm{G}_1,\ldots,\bm{G}_\ell)\) is a \((1/3)\)-error probabilistic polynomial for \(\OR_s\).
		
		Note that by Theorem~\ref{thm:BRST}, we can choose $\bm{O}$ to have degree at most $O(\log \ell).$ Further, by Fact~\ref{fac:pdeg}, we have $\pdeg(\bm{g}_i)\leq \pdeg(g) \leq \pdeg(f)$ for each $i\in [\ell]$. In particular, this implies that we can choose $\bm{G}_i$ to have degree $O(\pdeg(f)\cdot\log \ell)$ for each $i\in [\ell]$. This yields
		\[
		\pdeg(\OR_s) \leq \pdeg(f)\cdot O(\log \ell)^2 = \pdeg(f)\cdot O( (\log \log s)^2) = \pdeg(f) \cdot (\log s)^{o(1)}.
		\]
		As $\pdeg(\OR_s) \geq (\log s)^{c-o(1)}$ by assumption, we get the desired lower bound on $\pdeg(f).$
	\end{proof}
	
	It remains to prove Lemma~\ref{lem:OR-redn}, which we do now. 
	
	\begin{proof}[Proof of Lemma~\ref{lem:OR-redn}]
		We will only use restrictions $g'$ of $g$ obtained by setting some inputs of $g$ to $0$. The basic observation~\cite{BRS,Tarui} is the following. For any such restriction $g' = g_\rho$, the function $g'$ always agrees with $\OR_s$ at the all-zero input. Moreover, if $a$ is non-zero and has weight $t > 0$, then $g'(a) = \OR_s(a) = 1$ as long as exactly $t-1$ of the variables that are set to $1$ in $a$ are fixed to $0$ by $\rho$ (this follows from the fact that $g$ accepts any input of weight exactly $1$). While we cannot always choose a \emph{single} restriction that does this for all possible $a$, it is possible to choose a \emph{small number of restrictions} randomly such that for each non-zero $a$, at least one of them is guaranteed to work with high probability. We now see the details.
		
		
		For $i\in [\log s],$ let $\mc{D}_i$ be the distribution over subsets of $[s]$ where we pick each element independently to be in the set with probability $2^{-i}.$ For a (constant) parameter $p$ to be chosen later, let $\bm{S}_1^i,\ldots,\bm{S}_p^i$ be independent random subsets picked from distribution $\mc{D}_i.$ Each such set $\bm{S}_j^i$ is associated with the restriction $\bm{\rho}_j^i$ where each variable is set to $0$ if it does not belong to $\bm{S}_j^i$, and left alive (i.e. set to $*$) otherwise.  Note that $\bm{g}_j^i:= g_{\bm{\rho}_j^i}$ is a random restriction of $g$. Also observe that the total number of such restrictions is $\ell := p\log s = O(\log s).$ The final probabilistic representation is the $\OR$ of all these $\bm{g}_j^i$s.
		
		We now prove correctness. Consider any $a\in \{0,1\}^s$. The case when $a = 0^s$ is easy, as each $\bm{g}_j^i$ is obtained by setting some inputs of $g$ to $0$ and hence $\bm{g}_j^i (a) = 0$ with probability $1$.  The same is therefore true for the OR of these functions.
		
		Now assume that $a\neq 0$. Thus $|a| = t \in [s].$ Fix $i\in [\log s]$ such that $t\in (2^{i-1},2^i]$. We will show that, with probability at least $0.9$, some $\bm{g}_j^i$ evaluates to $1$. This will finish the proof.
		
		To see this, let $S\subseteq [s]$ be the set of coordinates where $a$ takes value $1$. Note that $\bm{g}_j^i(a) = g(\bm{b}_j^i )$ where $\bm{b}_j^i$ denotes the indicator vector of $\bm{S}_j^i\cap S$. As $g(b) = 1$ for any input $b$ of weight $1$, we see that $\bm{g}_j^i(a) = 1$ if $|\bm{S}_j^i \cap S| = 1$ . Hence, we have
		\begin{equation}
			\label{eq:OR-redn}
			\prob{\bm{g}_1^i,\ldots,\bm{g}_p^i}{\bm{g}_1^i(a) =\cdots = \bm{g}_p^i(a) = 0} \leq \mathop{\mathrm{Pr}}_{\bm{S}_1^i,\ldots,\bm{S}_p^i}\bigg[\bigwedge_{j=1}^p |\bm{S}_j^i\cap S| \neq 1\bigg] = \prod_{j=1}^p \prob{\bm{S}_j^i}{|\bm{S}_j^i\cap S| \neq 1},
		\end{equation} 
		where the last equality follows from the independence of the $\bm{S}_j^i$s.
		
		Finally, note that for any $j$,
		\begin{align*}
			\prob{\bm{S}_j^i}{|\bm{S}_j^i \cap S|=1} = \sum_{k\in S} \mathop{\mathrm{Pr}}_{\bm{S}_j^i}\bigg[k\in \bm{S}_j^i \wedge \bigwedge_{k'\in S\setminus k} k'\not\in \bm{S}_j^i\bigg]= t \cdot \frac{1}{2^i}\cdot \left(1-\frac{1}{2^i}\right)^{t-1} \geq \frac{1}{2}\cdot\left(1-\frac{1}{2^i}\right)^{2^i-1} \geq \frac{1}{2e},
		\end{align*}
		where the first inequality follows from the fact that $t\in (2^{i-1},2^i]$ and the second from the standard fact that $(1-1/n)^{n-1} \geq 1/e.$ Plugging the above into (\ref{eq:OR-redn}), we get 
		\[
		\prob{\bm{g}_1^i,\ldots,\bm{g}_p^i}{\bm{g}_1^i(a) =\cdots = \bm{g}_p^i(a) = 0}  \leq \left(1-\frac{1}{2e}\right)^p \leq \frac{1}{10},
		\]
		for a large enough constant $p$. In particular, for this $p$, the probability that $\OR_\ell(\bm{g}^i_j:i\in [\ell], j\in [p])$ evaluates to $0$ is at most $1/10$, completing the proof of the lemma.
	\end{proof}

	\subsection{The case of small sensitivity}
	
	We prove the following lemma.
	
	\begin{lemma}
		\label{lem:lbd-low-sens}
		Let $f:\{0,1\}^n\rightarrow \{0,1\}$ be a function of sensitivity at most $s$ that depends on all its $n$ variables. Then, we have $\pdeg(f) = \Omega\left(\frac{\log(n/s^{O(1)})}{\log s}\right)$.
	\end{lemma}
	
	The proof of the lemma is in two steps. In the first step, we use a counting argument to prove a lower bound on the probabilistic degrees of random functions $\bm{F}:\{0,1\}^m \rightarrow \{0,1\}$ which are chosen from a distribution such that for a large subset $X\subseteq \{0,1\}^m$, the random variables $\{\bm{F}(x)\ |\ x\in X\}$ are independently and uniformly chosen random bits. In the second step, we show how any $f$ as in the statement of Lemma~\ref{lem:lbd-low-sens} can be randomly restricted to a random function $\bm{F}$ where the lower bound for random functions applies. 
	
	We now state the lower bound for random functions and use it to prove Lemma~\ref{lem:lbd-low-sens}. The lower bound for random functions uses  fairly standard ideas and is proved in the appendix (Section~\ref{sec:app}).
	
	\begin{lemma}[Random function lower bound]
		\label{lem:rand-fn}
		The following holds for positive integer parameters $m, M$ and $d$ such that $M > m^{10d}$. Let $\bm{F}:\{0,1\}^m\rightarrow \{0,1\}$ be a random function such that for some $X\subseteq \{0,1\}^m$ with $|X| = M$, the random variables $(\bm{F}(x))_{x\in X}$ are independent and uniformly distributed random bits. Then, we have
		\[
		\prob{\bm{F}}{\pdeg_{1/10}(\bm{F}) \leq d} < \frac{1}{10}.
		\]
	\end{lemma}
	
	Let us see how to use Lemma~\ref{lem:rand-fn} to prove Lemma~\ref{lem:lbd-low-sens}. This proof again breaks into two smaller steps. 
	\begin{enumerate}
		\item[Step 1:] Show that, after a projection, $f$ turns into something similar to an addressing function, that we will call a \emph{Pseudoaddressing function.}
		\item[Step 2:] Show that any pseudoaddressing function has large probabilistic degree.
	\end{enumerate}
	As any projection $g$ of $f$ satisfies $\pdeg(g) \leq \pdeg(f)$ (Fact~\ref{fac:pdeg}), the above implies a lower bound on $\pdeg(f)$, hence proving Lemma~\ref{lem:lbd-low-sens}.
	
	To make the above precise, we need the following definition. We say that a function $g:\{0,1\}^{r+t}\rightarrow \{0,1\}$ is an \emph{$(r,t)$-Pseudoaddressing function} if the input variables to $g$ can be partitioned into two sets $Y = \{y_1,\ldots,y_r\}$ and $Z = \{z_1,\ldots,z_t\}$ and $g$ can be computed by a decision tree $T$ with the following properties.
	\begin{enumerate}
		\item[(P1)] For each $z_j\in Z$, there are two root-to-leaf paths $\pi_j^0$ and $\pi_j^1$ in $T$ that diverge at a node labeled $z_j$ and lead to outputs $0$ and $1$ respectively.
		\item[(P2)] All the other nodes on these paths are labeled by variables in $Y$, and further these variables take the same values on both paths. In particular, $\pi_j^0$ and $\pi_j^1$ differ only on the value of $z_j$.
	\end{enumerate}
	
	\begin{example}
		Consider the standard Addressing function $\mathrm{Addr}_r$ on $n= r+2^r$ variables as defined in Section~\ref{sec:pf-outline}. This function is an $(r,2^r)$-pseudoaddressing function as it can be computed by a decision tree of depth $r+1$, which first queries all the addressing variables to determine $a\in \{0,1\}^r$ and then queries and outputs the value of $z_a$ (the two computational paths querying $z_a$ give the desired root-to-leaf paths required in the definition above).
	\end{example}
	
	In analogy with the Addressing function, given an $(r,t)$-pseudoaddressing function as above, we refer to the variables in $Y$ as the addressing variables and the variables in $Z$ as the addressed variables.
	
	The two steps of the proof as outlined above can now be formalized as follows.
	
	\begin{claim}
		\label{clm:pseudo-proj}
		Let $f$ be as in the statement of Lemma~\ref{lem:lbd-low-sens}. Then, there exist $r\leq s^{O(1)}$ and $t \geq n/s^{O(1)}$ and a projection $\nu:[n]\rightarrow [r+t]$ such that $g = f|_\nu$ is an $(r,t)$-pseudoaddressing function.
	\end{claim}
	
	\begin{claim}
		\label{clm:pseudo-lbd}
		Let $g$ be any $(r,t)$-pseudoaddressing function. Then, $\pdeg(g) = \Omega(\log t/\log r).$
	\end{claim}
	
	As noted above, the above claims immediately imply Lemma~\ref{lem:lbd-low-sens}. We now prove these claims.
	
	\begin{proof}[Proof of Claim~\ref{clm:pseudo-proj}]
		We will first outline how to isolate a set of $n/\poly(s)$ variables that will (almost) be the set of addressed variables. A projection will then be applied to the remaining variables to create the pseudoaddressing function. Let us now see the details.
		
		By Theorem~\ref{thm:Huang}, we know that $f$ has a decision tree $T_f$ of depth $d\leq \poly(s)$. Fix such a tree $T_f$ of \emph{minimum size}, i.e. with the smallest possible number of leaves. Let $V=\{x_1,\ldots,x_n\}$ denote the input variables of $f$. 

		Any variable $x_i\in V$ must be queried somewhere in the tree $T_f$, as $f$ depends on all its input variables by assumption. Fix any occurrence of this variable in the decision tree $T_f$, and let $w$ denote the node of $T_f$ corresponding to this query.  (Refer to Figure~\ref{fig:main} (a) for an illustration.) Let $\pi_i$ denote the path from the root of $T_f$ to $w$ and let $T_0$ and $T_1$ be the subtrees rooted at the left and right children of $w$. The decision trees $T_0$ and $T_1$ both compute functions of the $n'<n$ Boolean variables not queried in $\pi_i$. Note that these decision trees compute \emph{distinct} functions since otherwise the query made at the vertex $w$ is unnecessary, and a smaller decision tree than $T_f$ can be obtained by replacing the subtree rooted at $w$ by $T_0$ or by $T_1$. This contradicts the minimality of the size of $T_f$.
		
		\begin{figure}[t!]
			\centering
			\begin{tikzpicture}[scale=0.8]
				\draw[ForestGreen,ultra thick] (0,0)--(-5+2.5,-2);
				\node[xshift=25,yshift=30] at (-5+2.5,-2) {\textcolor{ForestGreen}{\(\pi_8\)}};
				\draw (0,0)--(5-2.5,-2);
				
				\draw (-5+2.5,-2)--(-6+2.5,-4);
				\draw[ForestGreen,ultra thick] (-5+2.5,-2)--(-4+2.5,-4);
				\draw (5-2.5,-2)--(4-2.5,-4);
				\draw (5-2.5,-2)--(6-2.5,-4);
				
				\draw (-7+1+2.5,-4)--(-7.5+1+2.5,-6);
				\draw (-7+1+2.5,-4)--(-6-0.5+1+2.5,-6);
				\draw[ForestGreen,ultra thick] (-3-1+2.5,-4)--(-4+0.5-1+2.5,-6);
				\draw (-3-1+2.5,-4)--(-2-0.5-1+2.5,-6);
				\draw (3+1-2.5,-4)--(2+0.5+1-2.5,-6);
				\draw (3+1-2.5,-4)--(4-0.5+1-2.5,-6);
				\draw (7-1-2.5,-4)--(6+0.5-1-2.5,-6);
				\draw (7-1-2.5,-4)--(7.5-1-2.5,-6);
				
				\draw (-7.5+1+2.5,-6)--(-7.5+1+2.5,-8);
				\draw (-7.5+1+2.5,-6)--(-6.5-0.4-0.2+1+2.5,-8);
				\draw[Cyan,ultra thick] (-4+0.5-1+2.5,-6)--(-4.5+0.2+0.5+0.1-1+2.5-0.5,-8);
				\node[xshift=-16,yshift=-20] at (-4+0.5-1+2.5,-6) {\textcolor{Cyan}{\(\pi_w^1\)}};
				\draw[Cyan,ultra thick] (-4+0.5-1+2.5,-6)--(-3.5-0.2+0.5-0.1-1+2.5,-8);
				\node[xshift=10,yshift=-20] at (-4+0.5-1+2.5,-6) {\textcolor{Cyan}{\(\pi_w^0\)}};
				\draw (-2-0.5-1+2.5,-6)--(-2.5+0.2-0.5+0.1-1+2.5+0.2,-8);
				\draw (-2-0.5-1+2.5,-6)--(-1.5-0.2-0.5-0.1-1+2.5+0.2,-8);
				\draw (2+0.5+1-2.5,-6)--(1.5+0.2+0.5+0.1+1-2.5,-8);
				\draw (2+0.5+1-2.5,-6)--(2.5-0.2+0.5-0.1+1-2.5,-8);
				\draw (6+0.5-1-2.5,-6)--(5.5+0.2+0.5+0.1-1-2.5,-8);
				\draw (6+0.5-1-2.5,-6)--(6.5-0.2+0.5-0.1-1-2.5+0.5,-8);
				
				\draw (-4.5+0.2+0.5+0.1-1+2.5-0.5,-8)--(-4.5+0.2+0.5+0.1-1+2.5-0.5-0.2,-10);
				\draw[Cyan,ultra thick] (-4.5+0.2+0.5+0.1-1+2.5-0.5,-8)--(-4.5+0.2+0.5+0.1-1+2.5-0.5+0.2,-10);
				\draw (-3.5-0.2+0.5-0.1-1+2.5,-8)--(-3.5-0.2+0.5-0.1-1+2.5-0.2,-10);
				\draw[Cyan,ultra thick] (-3.5-0.2+0.5-0.1-1+2.5,-8)--(-3.5-0.2+0.5-0.1-1+2.5+0.2,-10);
				\draw (6.5-0.2+0.5-0.1-1-2.5+0.5,-8)--(6.5-0.2+0.5-0.1-1-2.5+0.5-0.2,-10);
				\draw (6.5-0.2+0.5-0.1-1-2.5+0.5,-8)--(6.5-0.2+0.5-0.1-1-2.5+0.5+0.2,-10);

				\vertexcirc{0}{0}{\(x_1\)}{}{Magenta,thick}
				
				\vertexcirc{-5+2.5}{-2}{\(x_2\)}{}{Magenta,thick}
				\vertexcirc{5-2.5}{-2}{\(x_2\)}{}{}
				
				\vertexcirc{-6+2.5}{-4}{\(x_5\)}{}{}
				\vertexcirc{-4+2.5}{-4}{\(x_3\)}{}{Magenta,thick}
				\vertexcirc{4-2.5}{-4}{\(x_6\)}{}{}
				\vertexcirc{6-2.5}{-4}{\(x_7\)}{}{}
				
				\vertexcirc{-7.5+1+2.5}{-6}{\(x_3\)}{}{}
				\vertexsq{-6.5+1+2.5}{-6}{0}{}{}
				\vertexcirc{-3.5-1+2.5}{-6}{\(x_8\)}{}{blue,thick}
				\node[xshift=-10,yshift=12] at (-3.5-1+2.5,-6) {\(w\)};
				\vertexcirc{-2.5-1+2.5}{-6}{\(x_9\)}{}{}
				\vertexcirc{2.5+1-2.5}{-6}{\(x_3\)}{}{}
				\vertexsq{3.5+1-2.5}{-6}{0}{}{}
				\vertexcirc{6.5-1-2.5}{-6}{\(x_4\)}{}{}
				\vertexsq{7.5-1-2.5}{-6}{1}{}{}
				
				\vertexsq{-7.5+1+2.5}{-8}{0}{}{}
				\vertexsq{-6.5-0.4-0.2+1+2.5}{-8}{1}{}{}
				\vertexcirc{-4.5+0.2+0.5+0.1-1+2.5-0.5}{-8}{\(x_4\)}{}{Magenta,thick}
				\vertexcirc{-3.5-0.2+0.5-0.1-1+2.5}{-8}{\(x_4\)}{}{Magenta,thick}
				\vertexsq{-2.5+0.2-0.5+0.1-1+2.5+0.2}{-8}{0}{}{}
				\vertexsq{-1.5-0.2-0.5-0.1-1+2.5+0.2}{-8}{1}{}{}
				\vertexsq{1.5+0.2+0.5+0.1+1-2.5}{-8}{0}{}{}
				\vertexsq{2.5-0.2+0.5-0.1+1-2.5}{-8}{1}{}{}
				\vertexsq{5.5+0.2+0.5+0.1-1-2.5}{-8}{0}{}{}
				\vertexcirc{6.5-0.2+0.5-0.1-1-2.5+0.5}{-8}{\(x_{10}\)}{}{}
				
				\vertexsq{-4.5+0.2+0.5+0.1-1+2.5-0.5-0.2}{-10}{0}{}{}
				\vertexsq{-4.5+0.2+0.5+0.1-1+2.5-0.5+0.2}{-10}{1}{}{}
				\vertexsq{-3.5-0.2+0.5-0.1-1+2.5-0.2}{-10}{1}{}{}
				\vertexsq{-3.5-0.2+0.5-0.1-1+2.5+0.2}{-10}{0}{}{}
				\vertexsq{6.5-0.2+0.5-0.1-1-2.5+0.5-0.2}{-10}{0}{}{}
				\vertexsq{6.5-0.2+0.5-0.1-1-2.5+0.5+0.2}{-10}{1}{}{}
				
				\node at (0,-11) {(a)\quad The tree $T_f$};
			\end{tikzpicture}
			\hspace{1cm}
			\begin{tikzpicture}[scale=0.8]
				\draw[ForestGreen,ultra thick] (0,0)--(-5+2.5,-2);
				\node[xshift=25,yshift=30] at (-5+2.5,-2) {\textcolor{ForestGreen}{\(\pi_8\)}};
				\draw (0,0)--(5-2.5,-2);
				
				\draw (-5+2.5,-2)--(-6+2.5,-4);
				\draw[ForestGreen,ultra thick] (-5+2.5,-2)--(-4+2.5,-4);
				\draw (5-2.5,-2)--(4-2.5,-4);
				\draw (5-2.5,-2)--(6-2.5,-4);
				
				\draw (-7+1+2.5,-4)--(-7.5+1+2.5,-6);
				\draw (-7+1+2.5,-4)--(-6-0.5+1+2.5,-6);
				\draw[ForestGreen,ultra thick] (-3-1+2.5,-4)--(-4+0.5-1+2.5,-6);
				\draw (-3-1+2.5,-4)--(-2-0.5-1+2.5,-6);
				\draw (3+1-2.5,-4)--(2+0.5+1-2.5,-6);
				\draw (3+1-2.5,-4)--(4-0.5+1-2.5,-6);
				\draw (7-1-2.5,-4)--(6+0.5-1-2.5,-6);
				\draw (7-1-2.5,-4)--(7.5-1-2.5,-6);
				
				\draw (-7.5+1+2.5,-6)--(-7.5+1+2.5,-8);
				\draw (-7.5+1+2.5,-6)--(-6.5-0.4-0.2+1+2.5,-8);
				\draw[Cyan,ultra thick] (-4+0.5-1+2.5,-6)--(-4.5+0.2+0.5+0.1-1+2.5-0.5,-8);
				\node[xshift=-16,yshift=-20] at (-4+0.5-1+2.5,-6) {\textcolor{Cyan}{\(\pi_w^1\)}};
				\draw[Cyan,ultra thick] (-4+0.5-1+2.5,-6)--(-3.5-0.2+0.5-0.1-1+2.5,-8);
				\node[xshift=10,yshift=-20] at (-4+0.5-1+2.5,-6) {\textcolor{Cyan}{\(\pi_w^0\)}};
				\draw (-2-0.5-1+2.5,-6)--(-2.5+0.2-0.5+0.1-1+2.5+0.2,-8);
				\draw (-2-0.5-1+2.5,-6)--(-1.5-0.2-0.5-0.1-1+2.5+0.2,-8);
				\draw (2+0.5+1-2.5,-6)--(1.5+0.2+0.5+0.1+1-2.5,-8);
				\draw (2+0.5+1-2.5,-6)--(2.5-0.2+0.5-0.1+1-2.5,-8);
				\draw (6+0.5-1-2.5,-6)--(5.5+0.2+0.5+0.1-1-2.5,-8);
				\draw (6+0.5-1-2.5,-6)--(6.5-0.2+0.5-0.1-1-2.5,-8);
				
				\draw (-4.5+0.2+0.5+0.1-1+2.5-0.5,-8)--(-4.5+0.2+0.5+0.1-1+2.5-0.5-0.2,-10);
				\draw[Cyan,ultra thick] (-4.5+0.2+0.5+0.1-1+2.5-0.5,-8)--(-4.5+0.2+0.5+0.1-1+2.5-0.5+0.2,-10);
				\draw (-3.5-0.2+0.5-0.1-1+2.5,-8)--(-3.5-0.2+0.5-0.1-1+2.5-0.2,-10);
				\draw[Cyan,ultra thick] (-3.5-0.2+0.5-0.1-1+2.5,-8)--(-3.5-0.2+0.5-0.1-1+2.5+0.2,-10);

				\vertexcirc{0}{0}{\(y_1\)}{}{Magenta,thick}
				
				\vertexcirc{-5+2.5}{-2}{\(y_2\)}{}{Magenta,thick}
				\vertexcirc{5-2.5}{-2}{\(y_2\)}{}{Magenta,thick}
				
				\vertexcirc{-6+2.5}{-4}{\(z_1\)}{}{blue,thick}
				\vertexcirc{-4+2.5}{-4}{\(y_3\)}{}{Magenta,thick}
				\vertexcirc{4-2.5}{-4}{\(z_2\)}{}{blue,thick}
				\vertexcirc{6-2.5}{-4}{\(z_3\)}{}{blue,thick}
				
				\vertexcirc{-7.5+1+2.5}{-6}{\(y_3\)}{}{Magenta,thick}
				\vertexsq{-6.5+1+2.5}{-6}{0}{}{}
				\vertexcirc{-3.5-1+2.5}{-6}{\(z_4\)}{}{blue,thick}
				\node[xshift=-10,yshift=12] at (-3.5-1+2.5,-6) {\(w\)};
				\vertexcirc{-2.5-1+2.5}{-6}{\(z_5\)}{}{blue,thick}
				\vertexcirc{2.5+1-2.5}{-6}{\(y_4\)}{}{Magenta,thick}
				\vertexsq{3.5+1-2.5}{-6}{0}{}{}
				\vertexcirc{6.5-1-2.5}{-6}{\(y_4\)}{}{Magenta,thick}
				\vertexsq{7.5-1-2.5}{-6}{1}{}{}
				
				\vertexsq{-7.5+1+2.5}{-8}{0}{}{}
				\vertexsq{-6.5-0.4-0.2+1+2.5}{-8}{1}{}{}
				\vertexcirc{-4.5+0.2+0.5+0.1-1+2.5-0.5}{-8}{\(y_4\)}{}{Magenta,thick}
				\vertexcirc{-3.5-0.2+0.5-0.1-1+2.5}{-8}{\(y_4\)}{}{Magenta,thick}
				\vertexsq{-2.5+0.2-0.5+0.1-1+2.5+0.2}{-8}{0}{}{}
				\vertexsq{-1.5-0.2-0.5-0.1-1+2.5+0.2}{-8}{1}{}{}
				\vertexsq{1.5+0.2+0.5+0.1+1-2.5}{-8}{0}{}{}
				\vertexsq{2.5-0.2+0.5-0.1+1-2.5}{-8}{1}{}{}
				\vertexsq{5.5+0.2+0.5+0.1-1-2.5}{-8}{0}{}{}
				\vertexsq{6.5-0.2+0.5-0.1-1-2.5}{-8}{1}{}{}
				
				\vertexsq{-4.5+0.2+0.5+0.1-1+2.5-0.5-0.2}{-10}{0}{}{}
				\vertexsq{-4.5+0.2+0.5+0.1-1+2.5-0.5+0.2}{-10}{1}{}{ForestGreen,very thick}
				\vertexsq{-3.5-0.2+0.5-0.1-1+2.5-0.2}{-10}{1}{}{}
				\vertexsq{-3.5-0.2+0.5-0.1-1+2.5+0.2}{-10}{0}{}{ForestGreen,very thick}
				
				\node at (0,-11) {(b)\quad The tree $T$ (after projection)};
			\end{tikzpicture}
			\caption{\small The decision tree on the left computes a truly $10$-variate function $f(x_1,\ldots,x_{10})$. The paths obtained by concatenating $\pi_8$ with $\pi_w^0$ and $\pi_w^1$ are consistent with each other except for the value of $x_8$, the variable queried at node $w$. After a projection $\nu:[10]\rightarrow [9]$ defined by $\nu(i) = i$ for $i\leq 9$ and $\nu(10) = 4$, we get a tree $T$, which computes a $(4,5)$-pseudoaddressing function $g(y_1,\ldots,y_4,z_1,\ldots,z_5)$. Note that each path in $T$ corresponds to a path in $T_f$ but not every path in $T_f$ survives in $T$ (e.g. the path leading to $0$ through the node querying $x_{10}$ is pruned away, as it is inconsistent with $\nu$).}
			\label{fig:main}
		\end{figure}
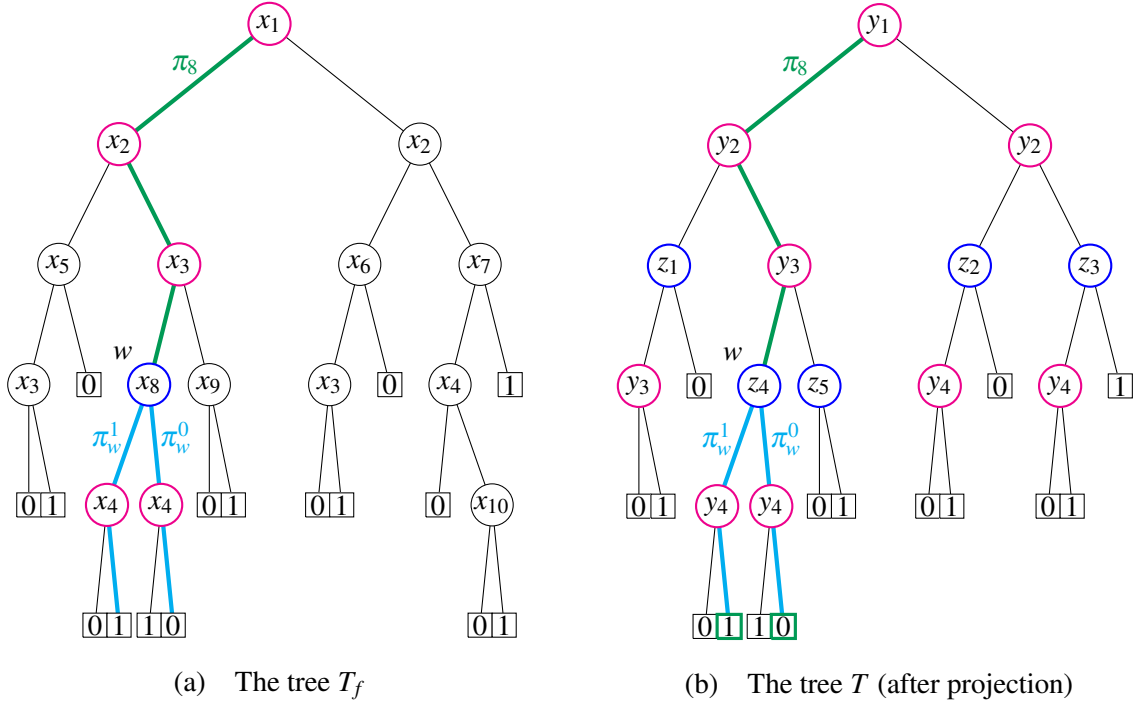
		
		Thus, $T_0$ and $T_1$ compute distinct functions. In particular, there is an input $a\in \{0,1\}^{n'}$ on which $T_0$ and $T_1$ have different outputs; w.l.o.g., assume $T_0$ and $T_1$ output $0$ and $1$ respectively on $a$. Let $\pi_w^0$ and $\pi_w^1$ be the root-to-leaf paths followed on the input $a$ in $T_0$ and $T_1$ respectively. Note that any variable queried on both $\pi_w^0$ and $\pi_w^1$ takes the same value on both paths, as both paths are consistent with the input $a$. (Again, see Figure~\ref{fig:main} (a) for an example.)
		
		Concatenating each of $\pi_w^0$ and $\pi_w^1$ with the path $\pi_i$ gives us two root-to-leaf paths $\pi_i^0$ and $\pi_i^1$ in $T$ such that
		\begin{enumerate}
			\item[P1$'$] The paths $\pi_i^0$ and $\pi_i^1$  diverge at the node $w$ (labelled by variable $x_i$) and lead to outputs $0$ and $1$ respectively.
			\item[P2$'$] The two paths agree on all variables other than $x_i$, i.e., any other variable that is queried on $\pi_i^0$ and $\pi_i^1$ takes the same value on both.
		\end{enumerate}
		
		We have such a pair of paths $\pi_i^0$ and $\pi_i^1$ for each $x_i\in V$. Let $P_i$ denote the set of all $j\neq i$ such that $x_j$ is queried on $\pi_i^0$ or on $\pi_i^1$. Note that $|P_i|\leq 2d.$ 
		
		We claim that we can choose a large subset $Z'\subseteq [n]$ such that for all  $i\in Z'$, the set $P_i$ does not contain any $j$ where $j\in Z'$. To see this, define a graph $G$ with vertex $[n]$ and edges between vertices distinct $i,j\in [n]$ if and only if $P_i$ contains $j$ or vice-versa. Since each $|P_i|\leq 2d$, it is clear that this graph has average degree at most $2d$. By Tur\'{a}n's theorem (see e.g.~\cite{alon-spencer}), this implies that $G$ has an independent set $Z'$ of size at least $n/4d.$ This set $Z'$ has the required property.
		
		We are now ready to show that the required projection $\nu$ exists. Let $r = 10d^2$ and let $\bm{\nu}':[n]\setminus Z' \rightarrow [r]$ be a random map (i.e. the image of each element of the domain is independently and uniformly chosen from $[r]$). We say that an $i\in Z'$ is \emph{good} if $\bm{\nu}'$ is $1$-$1$ on the set $P_i$.  Let $\bm{\mc{G}}$ be the set of all good $i$, with $\bm{t} := |\bm{\mc{G}}|$. Assume $\bm{\mc{G}} = \{{i_1}, \ldots, {i_{\bm{t}}} \}.$ We use this to define a random projection $\bm{\nu}:[n]\rightarrow [r+\bm{t}]$ by
		\[
		\bm{\nu}(i) = \left\{
		\begin{array}{ll}
			\bm{\nu}'(i) & \text{if $i\not\in Z'$,}\\
			1 & \text{if $i\in Z'\setminus \bm{\mc{G}}$, (here, any $k\in [r]$ will do)}\\
			r+j & \text{if $i\in \bm{\mc{G}}$ and $i = i_j.$}
		\end{array}\right.
		\]
		
		The random projection defines a random Boolean function $\bm{g}$ on $r+\bm{t}$ variables. We now show that, with positive probability, $\bm{g}$ is an $(r,n/\poly(s))$-pseudoaddressing function, where the first $r$ variables are the addressing variables. This will finish the proof. Note that the projection $\bm{\nu}$ applied to the tree $T_f$ also defines a random decision tree $\bm{T}$ computing $\bm{g}$. We will in fact show that $\bm{T}$ serves as a witness for the fact that $\bm{g}$ is an $(r,n/\poly(s))$-pseudoaddressing function (with positive probability).
		
		In fact, this happens whenever $\bm{t} = |\bm{\mc{G}}|$ is large enough. More precisely, note that 
		\begin{align*}
			\avg{\bm{\nu}'}{|Z'|-\bm{t}} &= \sum_{i\in Z'} \prob{\bm{\nu}'}{\text{$\bm{\nu}'$ is not $1$-$1$ on $P_i$}}  \leq \sum_{i\in Z'} \sum_{j\neq k\in P_i} \prob{\bm{\nu}'}{\bm{\nu}'(j) = \bm{\nu}'(k)}\leq \sum_{i\in Z'} |P_i|^2 \cdot \frac{1}{r} \leq |Z'|\cdot \frac{(2d)^2}{r} \leq \frac{|Z'|}{2}.
		\end{align*}
		In particular, there is a setting $\nu'$ of $\bm{\nu}'$ such that the corresponding set of good variables $|\bm{\mc{G}}|$ has size at least $|Z'|/2.$ Fix this $\nu'$ and let $\mc{G},t,\nu, g,T$ be the corresponding fixings of $\bm{\mc{G}}, \bm{t}, \bm{\nu}, \bm{g}, \bm{T}$ respectively. 
		
		We have $g = g(y_1,\ldots,y_r, z_1,\ldots,z_t)$. Observe that each root-to-leaf path of $T$ can be identified with a root-to-leaf path of $T_f$. Further, a path $\pi$ of $T_f$ survives in $T$ exactly when it is \emph{consistent w.r.t. $\nu$}, i.e., if two variables that are set to opposite values in $\pi$ are not mapped to the same variable by $\nu$ (see Figure~\ref{fig:main} (b) for an example). In particular, if a path $\pi$ has the property that the variables queried along $\pi$ are mapped injectively by $\nu$, then the path $\pi$ survives in $T$.
		
		This implies that for any good ${i_j}\in \mc{G}$, the corresponding paths $\pi_{{i_j}}^0$ and $\pi_{{i_j}}^1$ survive in $T$. Moreover, as the projection $\nu$ is injective on the entire set $P_{i_j}$, these paths continue to agree with each other on all variables except the variable $z_j$ queried at the point of their divergence. This gives both properties P1 and P2 stated above. As this holds for each ${i_j}\in \mc{G},$ we see that $g$ is indeed an $(r,t)$-pseudoaddressing function. Note that $r = 10d^2 \leq \poly(s)$ and $t\geq n/4d \geq n/\poly(s)$. Hence, we have proved the claim.
	\end{proof}
	
	\begin{proof}[Proof of Claim~\ref{clm:pseudo-lbd}]
		The proof is via a reduction to Lemma~\ref{lem:rand-fn}. 
		
		Let $g(y_1,\ldots,y_r,z_1,\ldots,z_t)$ be an $(r,t)$-pseudoaddressing function. Consider the random function $\bm{F}$ on $\{0,1\}^r$ obtained by setting the addressed variables $z_1,\ldots,z_t$ to $\bm{b}_1,\ldots,\bm{b}_t\in \{0,1\}$ chosen i.u.a.r.. We show that there is an $X\subseteq \{0,1\}^r$ of size $t$ such that the random variables $(\bm{F}(a): a\in X)$ are independent and uniformly distributed bits. Then, Lemma~\ref{lem:rand-fn} implies the statement of the claim.
		
		Let us see how $X$ is defined. Let $T$ be the decision tree guaranteed for $g$ by virtue of the fact that it is an $(r,t)$-pseudoaddressing function. Further, for any $z_j$, let $\pi_j^0$ and $\pi_j^1$ be the paths satisfying P1 and P2 above. By P2, we can fix a setting $a^{(j)}\in \{0,1\}^r$ to the $y$-variables that is consistent with both paths. We set $X = \{a^{(j)}\ |\ j\in [t]\}.$
		
		To analyze $\bm{F}(a^{(j)})$, note that setting the variables $z_1,\ldots,z_t$ to $\bm{b}_1,\ldots, \bm{b}_t$ in $T$ gives us a (random) decision tree $\bm{T}'$ that computes $\bm{F}.$ In particular, the path followed by $\bm{T}'$ on input $a^{(j)}$ is uniformly chosen among $\pi_j^{0}$ and $\pi_j^{1}$ depending on the value of $z_j$, and hence $\bm{F}(a^{(j)})$ is either $\bm{b}_j$ or $1-\bm{b}_j$ (exactly which depends on the value of $z_j$ that is consistent with $\pi_j^0$ and $\pi_j^1$). In either case, however, $\bm{F}(a^{(j)})$ is a uniformly chosen random bit depending only on $\bm{b}_j.$ Hence, the random variables $(\bm{F}(a^{(j)}): j\in [t])$ are independent and uniformly distributed. 
		
		Thus, Lemma~\ref{lem:rand-fn} implies that with positive probability, $\pdeg_{1/10}(\bm{F}) = \Omega(\log t/\log r).$ However, we know by Fact~\ref{fac:pdeg} that, as $\bm{F}$ is a restriction of $g$, $\pdeg_{1/10}(\bm{F}) \leq \pdeg_{1/10}(g).$ Hence, we obtain the same lower bound for $\pdeg_{1/10}(g).$ Finally, by error reduction (Fact~\ref{fac:pdeg}), the same lower bound (up to constant factors) holds for $\pdeg_{1/3}(g) = \pdeg(g).$
	\end{proof}

	\subsection{Finishing the proof of Theorem~\ref{thm:lbd}}
	
	Lemma~\ref{lem:lbd-low-sens} and Lemma~\ref{lemma:lbd-high-sens} imply that 
	\[
	\pdeg(f) =\Omega\left( \max\left\{(\log s)^{c-o(1)}, \frac{\log(n/s^{O(1)})}{\log s}\right\}\right)
	\]
	where $s$ denotes the sensitivity of $f$. The above is minimized for $s$ so that $(\log s)^{c+1} = \Theta(\log n)$ (note that this implies that $s = n^{o(1)}$). For this $s$, we get
	\[
	\pdeg(f) = \Omega( (\log n)^{c/(c+1) - o(1)}) \geq (\log n)^{c/(c+1)-o(1)},
	\]
	proving the theorem.
	
	\section{The Upper Bound: Proof of Theorem~\ref{thm:ubd}}
	
	The construction is motivated by and closely follows a construction of Ambainis and de Wolf~\cite{AW}, who used it to prove the existence of a truly $n$-variate Boolean function $f$ whose approximate degree is $O(\log n/\log \log n).$ The construction of~\cite{AW} uses the fact that the approximate degree of the $\OR_n$ function is $O(\sqrt{n})$~\cite{Grover}. Using our assumption that the probabilistic degree of the $\OR_n$ function is $(\log n)^{c+o(1)}$ we are able to prove a stronger degree upper bound for probabilistic degree. In particular, Theorem~\ref{thm:BRST} allows us to prove an unconditional upper bound of $(\log n)^{(1/2)+o(1)}$ on the probabilistic degree of some $n$-variable function.
	
	The construction is a variant of the Addressing function, where the addressing bits are replaced by elements of a larger alphabet $[s]$, which are themselves presented in an encoded form that allows them to be easily `decoded' by low-degree polynomials. More precisely, we construct the function as follows.
	
	\paragraph{Construction.} Let $s$ be a power of $2$ and let $H\subseteq \{0,1\}^s$ be the set of codewords of the Hadamard code. That is, assume $s = 2^t$ and identify elements of $\{0,1\}^s$ with functions $h:\{0,1\}^t \rightarrow \{0,1\}.$ Then $H$ consists of precisely those elements $h\in \{0,1\}^s$ such that $h$ is a linear function when considered as a mapping from $\F_2^s$ to $\F_2$ in the natural way. The set $H$ contains precisely $s$ elements, say $\{h_1,\ldots,h_s\}.$
	
	We define a Boolean function $f$ on $n = sr + s^r + 1$ bits as follows. Any input $a$ is parsed as 
	\[
	a = (g_1,\ldots,g_r, T, b)
	\]
	where $g_1,\ldots,g_r:\{0,1\}^t\rightarrow \{0,1\}$, $T:[s]^r \rightarrow \{0,1\}$ and $b$ is a single bit. We define $f$ by
	\[
	f(a) = \left\{
	\begin{array}{ll}
		T(i_1,\ldots,i_r) & \text{if $g_1,\ldots,g_r\in H$ and $g_1 = h_{i_1},\ldots,g_r = h_{i_r}$,}\\
		b & \text{otherwise}.
	\end{array}
	\right.
	\]
	
	\paragraph{Analysis.} We have
	\begin{equation}
		\label{eq:f-defn}
		f(g_1,\ldots,g_r,T,b)=\sum_{i_1,\ldots,i_r\in[s]}1(g_1=h_{i_1},\ldots,g_r=h_{i_r})\cdot T(i_1,\ldots,i_r)+(1-1(g_1,\ldots,g_r\in H))\cdot b.
	\end{equation}
	Here $1(\mc{E})$ for a Boolean predicate $\mc{E}$ takes the value $1$ when the Boolean predicate is satisfied and $0$ otherwise.
	
	The above implies, in particular, that the function $f$ is truly $n$-variate. To see this, say the variables of $f$ are
	\begin{itemize}
		\item $x_{j,\alpha}$ ($j\in [r], \alpha\in \{0,1\}^t$) encoding the entries of the truth tables of $g_1,\ldots,g_r,$ More formally, the variable $x_{j,\alpha}$ is set to $g_{j}(\alpha).$
		\item $y_{i_1,\ldots,i_r}$ encoding the entries of $T$, and 
		\item $y_0$ which gives the value of $b$.
	\end{itemize}	
	Any variable $x_{j,\alpha}$ is influential at an input $(g_1,\ldots,g_r, T, b)$ where $g_1,\ldots,g_r$ are $h_{i_1},\ldots,h_{i_r}\in H$ respectively, and $b \neq T(i_1,\ldots,i_r)$, which implies that flipping the value of $x_{j,\alpha}$ at this point changes the output from $T(i_1,\ldots,i_r)$ to $b$. The variable $y_{i_1,\ldots,i_r}$ is also influential at the same point. The variable $y_0$ is influential at any input where not all the $g_i$ are in $H$. Thus, we see that $f$ is indeed $n$-variate.
	
	Now, we will show an upper bound on $\pdeg(f).$ This will be done by constructing two polynomials.
	\begin{itemize}
		\item A $1/3$-error probabilistic polynomial $\bm{Q}(x_{j,\alpha}: j\in [r], \alpha\in \{0,1\}^t)$ for the Boolean function  $1(g_1,\ldots,g_r\in H)$.
		\item For each $i_1,\ldots,i_r\in [s]$, a polynomial $R_{i_1,\ldots,i_r}(x_{j,\alpha}: j\in [r], \alpha\in \{0,1\}^t)$ such that at input $(g_1,\ldots,g_r)\in H^r$, $R_{i_1,\ldots,i_r}(g_1,\ldots,g_r) = 1$ if $g_1 = h_{i_1},\ldots,g_r = h_{i_r}$, and $0$ otherwise. (In other words, $R_{i_1,\ldots,i_r}$ computes a $\delta$-function on inputs from $H^r$. Note that we do not claim anything if $(g_1,\ldots,g_r)\not\in H^r$.)
	\end{itemize}
	
	Given the above constructions, the following yields a probabilistic polynomial $\bm{P}$ for $f$.
	\begin{equation}
		\label{eq:final-ubd}
		\bm{P} = \bm{Q}\cdot \left(\sum_{i_1,\ldots,i_r\in [s]} R_{i_1,\ldots,i_r}\cdot y_{i_1,\ldots,i_r}\right) + (1-\bm{Q})\cdot y_0
	\end{equation}
	(The two copies of $\bm{Q}$ are chosen with the same randomness and are \emph{not} independent of each other.)
	To see that this works, fix any input $a = (g_1,\ldots,g_r, T,b)$. If $(g_1,\ldots,g_r)\in H^r,$ the term in the parenthesis evaluates to $T(i_1,\ldots,i_r)$ with probability $1$. Further, $\bm{Q}(g_1,\ldots,g_r)$ evaluates to $1$ with probability $2/3$. Hence, $\bm{P}(a) = T(i_1,\ldots,i_r) = f(a)$ with probability at least $2/3$. On the other hand, if  $(g_1,\ldots,g_r)\not\in H^r,$ then $\bm{Q}(g_1,\ldots,g_r)$ evaluates to $0$ with probability $2/3$. When this event occurs, the first summand evaluates to $0$ and the second summand evaluates to $b$. Hence, $\bm{P}(a) = b= f(a)$ with probability at least $2/3$.
	
	It remains  to construct the polynomials $\bm{Q}$ and $R_{i_1,\ldots,i_r}$. We start with $\bm{Q}.$ Recall that a function $g:\{0,1\}^t \rightarrow  \{0,1\}$ lies in $H$ when it is linear over $\F_2$, or equivalently if $g(\alpha\oplus \beta) \oplus g(\alpha)\oplus g(\beta) = 0$ for every $\alpha,\beta\in \{0,1\}^t.$ Thus, the condition that $g_1,\ldots,g_r\in H$ can be rewritten as
	\[
	\bigwedge_{j=1}^r \bigwedge_{\alpha,\beta\in \{0,1\}^t}(1\oplus g_j(\alpha\oplus \beta) \oplus g_j(\alpha)\oplus g_j(\beta) ).
	\]
	Let $q(z_1,z_2,z_3)$ be a constant-degree polynomial of $3$ Boolean variables that evaluates to $1\oplus z_1\oplus z_2\oplus z_3$. Then, the above can be rewritten as $\bigwedge_{j=1}^r \bigwedge_{\alpha,\beta\in \{0,1\}^t}q(g_j(\alpha),g_j(\beta),g_j(\alpha\oplus \beta)).$ Thus, we can define the probabilistic polynomial to be
	\[
	\bm{Q}(x_{j,\alpha}: j\in [r], \alpha\in \{0,1\}^t) = \bm{Q}_1(q(x_{j,\alpha},x_{j,\beta},x_{j,\alpha\oplus \beta}): j\in [r], \alpha,\beta\in \{0,1\}^t)
	\]
	where $Q_1$ is any probabilistic polynomial for the $\AND_{r2^{2t}} = \AND_{rs^2}$ function. By assumption, $\pdeg(\OR_{rs^2})$ and hence, by DeMorgan's laws, $\pdeg(\AND_{rs^2})$ is at most $\log(rs^2)^{c+o(1)} = (\log r + \log s)^{c+o(1)}.$
	
	We now see how to construct $R_{i_1,\ldots,i_r}$ for any fixed $i_1,\ldots,i_r\in [s]$. Recall the standard fact (see, e.g.~\cite{ODonnellbook}) that for $h_{i_1}\neq h_{i_2} \in H$, the functions $\hat{h}_{i_1},\hat{h}_{i_2}:\{0,1\}^t\rightarrow \{-1,1\}$ defined by
	\begin{align*}
		\hat{h}_{i_b}(\alpha) &= 1-2h_{i_b}(\alpha), &&\text{ for all $\alpha\in\{0,1\}^t$ and $b\in \{1,2\}$},
	\end{align*}
	are orthogonal to one another, i.e., $\sum_{\alpha} \hat{h}_{i_1}(\alpha) \hat{h}_{i_2}(\alpha) = 0$. Based on this observation, we define the polynomial as follows.
	\[
	R_{i_1,\ldots,i_r}(x_{j,\alpha}: j\in [r], \alpha\in \{0,1\}^t) = \frac{1}{s^r} \prod_{j=1}^r \left(\sum_{\alpha\in \{0,1\}^t}\hat{h}_{i_j}(\alpha) (1-2x_{j,\alpha})\right).
	\]
	Let us see that this polynomial has the desired properties. Consider input $(g_1,\ldots,g_r)\in H^r$. Assume $g_{j} = h_{i'_j}$ for each $j\in [r]$. Then, we have
	\begin{align*}
		R_{i_1,\ldots,i_r}(g_1,\ldots,g_r) &= \frac{1}{s^r} \prod_{j=1}^r \left(\sum_{\alpha\in \{0,1\}^t}\hat{h}_{i_j}(\alpha) (1-2h_{i'_j}(\alpha))\right)\\
		&= \frac{1}{s^r} \prod_{j=1}^r \left(\sum_{\alpha\in \{0,1\}^t}\hat{h}_{i_j}(\alpha) \hat{h}_{i'_j}(\alpha)\right)
	\end{align*}
	and the latter quantity can be seen to be $1$ if $i'_j = i_j$ for all $j\in [r]$ and $0$ otherwise. Thus, $R_{i_1,\ldots,i_r}$ behaves as stipulated. Note that $\deg(R_{i_1,\ldots,i_r}) = r.$
	
	This concludes the construction of the probabilistic polynomial for $f$. The degree of the polynomial thus constructed is at most $\deg(\bm{Q}) + \max_{i_1,\ldots,i_r}\deg(R_{i_1,\ldots,i_r})  = O((\log r + \log s)^{c+o(1)} + r) = O((\log s)^{c+o(1)} + r).$
	
	\paragraph{Parameters.} We set $r = (\log s)^c = t^c$. This gives a truly $n$-variate Boolean function on $n = O(s^r) = O(2^{t^{1+c}})$ variables with probabilistic degree $t^{c+o(1)} = (\log n)^{(c/(c+1)) + o(1)}.$
	
	
	\bibliographystyle{abbrv}
	\bibliography{pdeg-nvariate-references}
	
	\appendix
	
	\section{Proof of the Random function lower bound (Lemma~\ref{lem:rand-fn})}
	\label{sec:app}
	
	The proof is via a counting argument. 
	
	We start with a standard observation, which follows from a simple averaging argument. If $F:\{0,1\}^m\rightarrow \{0,1\}$ has $(1/10)$-error probabilistic degree $d$, then for any probability distribution $\mu$ over $\{0,1\}^m$, there is a polynomial $P$ of degree at most $d$ such that 
	\begin{equation}
		\label{eq:yao}
		\prob{a\sim \mu}{P(a)= F(a)}\geq \frac{9}{10}.
	\end{equation}
	Conversely, if there is a probability distribution $\mu$ such that (\ref{eq:yao}) does not hold for any polynomial of degree at most $d$, then $\pdeg(F) >d$. We will take the hard distribution to be the uniform distribution over $X$. 
	
	More precisely, call a function $g:X\rightarrow \{0,1\}$ \emph{bad} if there is a polynomial $P$ of degree at most $d$ that agrees with $g$ on at least $9|X|/10 = 9M/10$ points of $X$. Let $\mc{B}$ be the set of bad functions. The reasoning above tells us that
	\begin{equation}
		\label{eq:rand-fn}
		\prob{\bm{F}}{\pdeg_{1/10}(\bm{F}) \leq d} \leq \prob{\bm{F}}{\bm{F}|_X\in \mc{B}} = \frac{|\mc{B}|}{2^{M}}.
	\end{equation}
	where for the latter inequality we have used the fact that the random variables $(\bm{F}(x): x\in X)$ are independently and uniformly distributed. Hence, it will suffice to bound $|\mc{B}|$ to prove the lemma.
	
	To bound the size of $\mc{B},$ it will suffice to give a short encoding of each element of $\mc{B}.$ Fix any $g\in \mc{B}$ and a polynomial $P$ that agrees with $g$ on a set $X'\subseteq X$ such that $|X'| \geq 9M/10$. Note that $g$ can be specified by
	\begin{enumerate}
		\item The set $X'$.
		\item The set of values of $g$ on $X\setminus X'$ (in some pre-determined order).
		\item A polynomial $Q$ of degree at most $d$ that agrees with $g$ on $X'$ (specified as a list of coefficients of monomials).
	\end{enumerate}
	
	Note that the number of choices for $X'$ is at most $\binom{M}{\leq M/10}$, which is bounded by $2^{H(1/10)M},$ where $H(\cdot)$ denotes the binary entropy function. Further, the number of possibilities for $g$ on $X\setminus X'$ is at most $2^{|X\setminus X'|} \leq 2^{M/10}.$
	
	It remains to bound the number of possibilities for $Q$. A priori, it is not completely clear how to bound the number of $Q$ as the coefficients of $Q$ could be arbitrary real numbers. However, we note that if there is a polynomial $P$ that agrees with $g$ on $X'$, then there is also a $Q$ that satisfies this property, and furthermore, the coefficients of $Q$ are rational numbers of small bit complexity. 
	
	Formally, we will use the following lemma, which is an easy consequence of \cite[Corollary 3.2d]{Schrijver}.
	
	\begin{lemma}
		\label{lem:cramer}
		Consider a system of linear equations $Ax = b$ over the rational numbers, where $A$ is an $p\times q$ \emph{Boolean} matrix, and $b\in \{0,1\}^p.$ Then, if the system has a real solution, it has a rational solution that can be specified (as a list of numerator-denominator pairs in binary) by at most $10q^3$ bits.
	\end{lemma}
	
	To use the above lemma, consider the problem of finding a polynomial $Q$ of degree at most $d$ that agrees with $g$ at all points in $X'$. The coefficients of such a polynomial $Q$ solve a linear system of $p:=|X'|$ many linear equations in $q := \binom{m}{\leq d}$ variables. By the existence of the polynomial $P$, this system has a solution. Thus by Lemma~\ref{lem:cramer}, we know that there is a solution of bit-complexity at most $10q^3 \leq m^{4d} < M/10.$  Therefore, we may always choose $Q$ from the set $\mc{Q}$ of polynomials of bit-complexity (as specified above) at most $M/10$. Note that $|\mc{Q}|\leq 2^{M/10}$ by definition.
	
	Overall, this gives a complete specification of any given $g\in \mc{B}.$ More precisely, we have given a $1$-$1$ map $\tau:\mc{B}\rightarrow \mc{X}\times S\times \mc{Q}$, where $\mc{X}$ is the collection of subsets of $X$ of size at least $9M/10,$ $S$ is the set of Boolean tuples of length $M/10$, and $\mc{Q}$ is the set of polynomials of degree at most $d$ of bit-complexity at most $M/10.$ Hence, $|\mc{B}|\leq |\mc{X}|\cdot |S|\cdot |\mc{Q}| \leq 2^{M\cdot (H(1/10)+1/10+1/10)} \leq 2^{9M/10}.$ Plugging this into (\ref{eq:rand-fn}), we get
	\[
	\prob{\bm{F}}{\pdeg_{1/10}(\bm{F})\leq d} \leq \frac{2^{9M/10}}{2^M} < \frac{1}{10}.
	\]
	
	This finishes the proof of the lemma.

\end{document}